\documentclass[conference]{IEEEtran}

\makeatletter
\def\ps@headings{%
\def\@oddhead{}%
\def\@evenhead{}%
\def\@oddfoot{\mbox{}\scriptsize\rightmark \hfil \thepage}%
\def\@evenfoot{\scriptsize\thepage \hfil \leftmark\mbox{}}}
\makeatother

\usepackage{epsf}
\usepackage{times}
\usepackage{epsfig}
\usepackage{graphicx}
\usepackage{epstopdf}
\usepackage{amsmath}
\usepackage{amssymb,amsthm}
\usepackage{amsxtra}
\usepackage{here}
\usepackage{rawfonts}
\usepackage{times}
\usepackage{url}
\usepackage{cite}
\usepackage{subcaption}
\usepackage[usenames, dvipsnames]{color}
\usepackage{algorithm}

\usepackage{algorithmic}

\newtheorem{corollary}{\bf Corollary}
\newtheorem{theorem}{\bf Theorem}

\newtheorem{definition}{\bf Definition}
\newtheorem{remark}{Remark}
\newlength{\aligntop}
\newlength{\alignbot}
\makeatletter
\renewenvironment{align}{%
  \vspace{\aligntop}
  \start@align\@ne\st@rredfalse\m@ne
}{%
  \math@cr \black@\totwidth@
  \egroup
  \ifingather@
    \restorealignstate@
    \egroup
    \nonumber
    \ifnum0=`{\fi\iffalse}\fi
  \else
    $$%
  \fi
  \ignorespacesafterend%
  \vspace{\alignbot}\par\noindent
}

\IEEEoverridecommandlockouts
\begin{document}
\title{Hardware Trojan Detection Game: A Prospect-Theoretic Approach}
%\maketitle
\author{\IEEEauthorblockN{ Walid Saad$^1$, Anibal Sanjab$^1$, Yunpeng Wang$^2$, Charles Kamhoua$^3$, and Kevin Kwiat$^3$} \IEEEauthorblockA{\small
$^1$ Wireless@VT, Bradley Department of Electrical and Computer Engineering,  Virginia Tech, Blacksburg, VA, USA, \\ Emails: \url{{walids,anibals}@vt.edu}\\
$^2$ Electrical and Computer Engineering Department, University of Miami, Coral Gables, FL, USA, Email: \url{y.wang68@umiami.edu}\\
$^3$ Air Force Research Laboratory, Information Directorate, Cyber Assurance Branch, Rome, NY \\ Emails: \url{{charles.kamhoua.1,kevin.kwiat}@us.af.mil}\vspace{-0.8cm}
 }%
  \thanks{This research is supported by the  U.S. National Science Foundation under Grants CNS-1253731 and CNS-1406947 and by the Air Force Office of Scientific Research (AFOSR).}}
\date{}
\maketitle

\begin{abstract}
Outsourcing integrated circuit (IC) manufacturing to offshore foundries has grown exponentially in recent years. Given the critical role of ICs in the control and operation of vehicular systems and other modern engineering designs, such offshore outsourcing has led to serious security threats due to the potential of insertion of hardware trojans -- malicious designs that, when activated, can lead to highly detrimental consequences. In this paper, a novel game-theoretic framework is proposed to analyze the interactions between a hardware manufacturer, acting as attacker, and an IC testing facility, acting as defender. The problem is formulated as a noncooperative game in which the attacker must decide on the type of trojan that it inserts while taking into account the detection penalty as well as the damage caused by the trojan. Meanwhile, the resource-constrained defender must decide on the best testing strategy that allows optimizing its overall utility which accounts for both damages and the fines. The proposed game is based on the robust \emph{behavioral framework of prospect theory (PT)} which allows capturing the potential uncertainty, risk, and irrational behavior in the decision making of both the attacker and defender. For both, the standard rational expected utility (EUT) case and the PT case, a novel algorithm based on fictitious play is proposed and shown to converge to a mixed-strategy Nash equilibrium. For an illustrative case study, thorough analytical results are derived for both EUT and PT to study the properties of the reached equilibrium as well as the impact of key system parameters such as the defender-set fine. Simulation results assess the performance of the proposed framework under both EUT and PT and show that the use of PT will provide invaluable insights on the outcomes of the proposed hardware trojan game, in particular, and system security, in general.
\end{abstract}

\section{Introduction}
The past decade has witnessed unprecedented advances in the fabrication and design of integrated circuits (ICs). Indeed, ICs have become an integral component in many engineering domains ranging from transportation systems and critical infrastructures to robotics, communication, and vehicular systems~\cite{SURV00}. For instance, the vast advancements in vehicular systems designs have led to wide developments of vehicular electronics technologies and the proliferation of the integration of ICs in vehicular systems. These massive advances in IC design have also had many production implication. In particular, the flexibility of modern IC design coupled with its ease of manufacturing have led to the outsourcing of IC fabrication~\cite{HP00}. Such outsourcing allows a cost-effective production of the IC circuitry of many systems and critical infrastructures~\cite{SURV01,HP00,SURV02}. Moreover, the recent interest in the use of commercial off-the-shelf devices in both civilian and military systems has also constituted yet another motivation for outsourcing IC fabrication~\cite{COTS000}.

Relying on offshore foundries for IC manufacturing is a cost-effective way for mass production of microcircuits. However, such an outsourcing can lead to serious security threats. These threats are exacerbated when the ICs in question are deployed into critical applications such as vehicular systems, communication systems, power networks, transportation systems, or military applications. One such threat is that of the \emph{hardware trojan} insertion by IC manufacturers~\cite{SURV02,HT04,HT00,HT01,HT02,HT03}. A hardware trojan is a malicious design that can be introduced into an IC at manufacturing. The trojan lies inactive until it is activated by certain pre-set conditions when the IC is in use. Once activated, the trojan can lead to a circuit error which, in turn, can lead to detrimental consequences to the system in which the IC is used. The threat of serious malicious IC alterations via hardware trojans has become a major concern to governmental and private agencies, as well as to the military, transportation, and energy sectors~\cite{SURV00,SURV01,SURV02,HP00,HT00,HT01,HT02,HT03,HT04}. For instance, vehicular technologies are known to be one of the main potential targets for hardware trojans~\cite{SurveyHardwareTrojan,RiseofHardwareTroj,ReducingEncryptionOverhead,HuntSwitch}. Indeed, due to their significant reliance on microcontrollers, digital-signal processors, microprocessors, commercial-off-the-shelf parts, and integrated circuits which can come from a vast range of suppliers, vehicular systems can be a prime target to electronic manipulation attacks and the insertion of hardware trojans. For example, in a recent study of auto industry trends~\cite{PWC}, it was observed that electronic systems contribute to $90\%$ of automobile innovations and new features. In addition, new airborne systems and military fighters contain hundreds and thousands of chips~\cite{HuntSwitch} with a large number of suppliers, spread around the world, making them a vulnerable target to potential hardware trojan insertion.

Defending against hardware trojans and detecting their presence face many challenges that range from circuit testing and design to economic and contractual issues~\cite{SURV00,SURV01,SURV02,HP00,HT00,HT01,HT02,HT03,HT04,HT05,HT06,HT07}. The majority of these works~\cite{HT00,HT01,HT02,HT03,HT04,HT05,HT06,HT07} focuses on IC and hardware-level testing procedures used to activate or detect hardware trojans. This literature also highlights a key limitation in testing for hardware trojans: there exists a resource limitation that prevents testing for all possible types of hardware trojans within a given circuit. While interesting, most of these existing works do not take into account the possible strategic interactions that can occur between the two entities involved in hardware trojan detection: the manufacturer of the IC and the recipient, such as the governmental agencies or companies that are buying the ICs. Indeed, on the one hand, the manufacturer (viewed as an attacker) can strategically decide on which type of trojan to insert while taking into account possible testing strategies of the IC recipient. On the other hand, the agency (viewed as a defender), must decide on which testing process to use and for which trojans to test, given the possible trojan types that a manufacturer can introduce. This motivates the need for a mathematical framework that allows a better understanding of these strategic interactions between the two entities and their strategic behavior in order to anticipate the outcome of such interaction.

To this end, recently, a number of research works~\cite{TrustGamesTrojan2,TrustGamesTrojan,TrojanGameConf,TrojanGameJrnl} have focused on modeling the strategic interaction between a manufacturer and an agency (or client) in a hardware trojan insertion/detection setting using game theory. In particular, the works in~\cite{TrustGamesTrojan2} and~\cite{TrustGamesTrojan} propose a game-theoretic method to test the effectiveness of hardware trojan detection techniques. These works have pinpointed the advantages of using game theory for the development of better hardware trojan detection strategies. In addition, the authors in~\cite{TrojanGameConf} and~\cite{TrojanGameJrnl} studied a zero-sum game between a hardware trojan attacker and defender aiming at characterizing the best detection strategy that the defender can employ to face a strategic attacker which can insert one of many types of trojans. Despite being interesting, these works assume that the involved players always act with full rationality. However, as has been experimentally tested in~\cite{kahneman1979prospect} and~\cite{tversky1992advances}, when faced with risks and uncertainty (as in the case of security situations such as hardware trojan detection scenarios) humans tend to act in a subjective and sometimes irrational manner. The works in~\cite{TrojanGameConf,TrojanGameJrnl,TrustGamesTrojan2,TrustGamesTrojan} do not take into account this subjectivity which would significantly impact the game-theoretic results and equivalently affect the optimal attack and defense strategies of the involved entities. As such, a fundamentally new approach is needed that incorporates this possible subjective behavior in the game-theoretic formulation in order to quantify and assess the impact of such subjectivity on the attacker's and defender's strategies as well as on the hardware trojan detection game's outcome.  

The main contribution of this paper is to propose a novel, game-theoretic framework to understand how the attacker and defender can interact in a hardware trojan detection game. We formulate the problem as a noncooperative zero-sum game in which the defender must select the trojan types for which it wishes to test while the attacker must select a certain trojan type to insert into the IC. In this game, the attacker aims to maximize the damage that it inflicts on the defender via the trojan-infected IC while the defender attempts to detect the trojan and, subsequently, impose a penalty that would limit the incentive of the attacker to insert a trojan. One key feature of the proposed game is that it allows, based on the emerging framework of prospect theory (PT), capturing the subjective behavior of the attacker and defender when choosing their strategies under uncertainty and risk that accompany the hardware trojan detection decision making processes. This uncertainty and risk stem from the lack of information that the attacker and defender have on one another as well as from the tragic consequences on the attacker and defender that are associated, respectively, with a successful or unsuccessful detection of the trojan. Moreover, such a subjective behavior can originate from the personality traits of the humans involved (e.g, system administrators at the defense side and hackers at the manufacturer's side) which guide their tendency of being risk seeking or risk averse. Using PT enables studying how the attacker and defender can make their decisions based on subjective perceptions on each others' possible strategies and the accompanying gains and losses. To our best knowledge, this is the first paper that applies tools from PT to better understand the outcomes of such a security game. 
Indeed, although game theory has been a popular tool for network security (see survey in \cite{QZ00}), most existing works are focused on games in which all players are rational (one notable exception is in~\cite{MT00} which, however, focuses on resource allocation and does not address hardware trojan detection). Moreover, beyond some recent works on using PT for wireless networking~\cite{PTM00} and smart grid~\cite{WSM00}, no work seems to have investigated how PT can impact system security, in general, and trojan detection, in particular. To solve the game under both standards, rational expected utility theory (EUT) and PT, we propose an algorithm based on fictitious play that is shown to converge to a mixed-strategy Nash equilibrium of the game. Then, for an illustrative numerical case study, we derive several analytical results on the equilibrium properties and the impact of the fine (i.e. penalty) on the overall outcome of the game. Simulation results show that PT provides insightful results on how uncertainty and risk can impact the overall outcome of a security game, in general, and a hardware trojan detection game in particular. The results show that deviations from rational EUT decision making can lead to unexpected outcomes for the game. Therefore, these results will provide guidelines for system designers to better understand how to counter hardware trojans and malicious manufacturers.

The rest of this paper is organized as follows: Section~\ref{sec:prob} presents
the system model and the formulation of a noncooperative game for hardware trojan detection. In Section~\ref{sec:pt}, we present a novel trojan detection framework based on PT while in Section~\ref{sec:algo} we devise an algorithm for solving the game. Analytical and simulation results are presented and analyzed in Section~\ref{sec:num} while conclusions are drawn in
Section \ref{sec:conc}.

\section{System Model and Game Formulation}\label{sec:prob}
\subsection{System Model}
Consider an IC manufacturer who produces ICs for different governmental agencies or companies. This manufacturer, hereinafter referred to as an ``attacker'', has an incentive to introduce hardware trojans to maliciously impact the cyber-infrastructure that adopts the produced IC. Such a trojan, when activated, can lead to errors in the circuit, potentially damaging the underlying system. Here, we assume that the attacker can insert one trojan $t$ from a set $\mathcal{T}$ of $T$ trojan types. Each trojan $t \in \mathcal{T}$ can lead to a certain damage captured by a positive real-number $V_t > 0$. 

Once the agency or company, hereinafter referred to as the ``defender'', receives the ICs, it can decide to test for one or more types of trojans. Due to the complexity of modern IC designs, it is challenging to develop test patterns that can be used to readily and quickly verify the validity of a circuit with respect to all possible trojan types. Particularly, the defender must spend ample resources if it chooses to test for all possible types of trojans. Such resources may be extremely costly. Thus, we assume that the defender can only choose a certain subset $\mathcal{A} \subset \mathcal{T}$ of trojan types for which to test, where the total number of trojans tested for is $|\mathcal{A}| < T$. The practical aspects for testing and verification of the circuit versus the subset of trojans $\mathcal{A}$ can follow existing approaches such as the scan chain approach developed in \cite{CR00}. We assume that such testing techniques are reliable and, thus, if the defender tests for the accurate type of trojan, this trojan can then be properly detected.

Here, if the defender tests for the right types of trojans that have been inserted in the circuit, then, the attacker will be penalized. This penalty is mathematically expressed by a fine $F_t$ if the trojan detected is of type $t$. The magnitude and severity of this penalty depends on the seriousness of the threat. Thus, this fine is a mathematical representation of the legal consequences of the detection of the induced threat on the manufacturing company and the involved personnel including the termination of the contract (highly damaging the reputation of the manufacturer) between the two parties as well as monetary penalties that the manufacturer is required to pay for the defender. 

Our key goal is to understand the interactions between the defender and attacker in such a hardware trojan detection scenario. In particular, it is of interest to devise an approach using which one can understand how the defender and attacker can decide on the types of trojans that they will test for or insert, respectively, and how those actions impact the overall damage on the system.  
Such an approach will provide insights on the optimal testing choices for the defender, given various possible actions that could be taken by the attacker.

\subsection{Noncooperative Game Formulation}
For the studied hardware trojan detection model, the decision of the defender regarding for which trojans to test is impacted by its perception of the potential decisions of an attacker regarding which type of trojan to insert and vice versa. Moreover, the choices by both attacker and defender will naturally determine whether any damages will be done to the system or whether any penalty must be imposed. Due to this coupling in the actions and objectives of the attacker and defender, the framework of noncooperative game theory~\cite{TB00} provides suitable analytical tools for modeling, analyzing, and understanding the decision making processes involved in the studied attacker-defender hardware trojan detection scenario.

To this end, we formulate a static \emph{zero-sum noncooperative game} in strategic form  $\Xi=\{\mathcal{N}, \{\mathcal{S}_i\}_{i\in\mathcal{N}}, \{u_i\}_{i\in\mathcal{N}}\}$ which is defined by its three main components: (i) the \emph{players} which are the attacker $a$ and the defender $d$ in the set\footnote{This two-player game formulation captures practical cases in which one system operator defends its system against trojan insertion while the system is considered not to be extremely vulnerable, in the sense that, most of the manufacturers are trusted while very few (in our case a single manufacturer) are malicious. Our generated results and proposed techniques can also form the basis for future works focusing on applications in which the existence of multiple attackers or multiple defenders represents a more practical case.} $\mathcal{N} := \{a,d\}$, (ii) the \emph{strategy} space $\mathcal{S}_i$ of each player $i \in \mathcal{N}$, and (iii) the \emph{utility function} $u_i$ of any player $i \in \mathcal{N}$.

For the attacker, the strategy space is simply the set of possible trojan types, i.e., $\mathcal{S}_a=\mathcal{T}$. Thus, an attacker can choose one type of trojans to insert in the circuit being designed or manufactured. For the defender, given the possibly large number of trojans that must be tested for, we assume that the defender can only choose to test for $K$ trojan types simultaneously. The actual value of $K$ would be determined exogeneously to the game via factors such as the resources available for the defender and the type of circuitry being tested. For a given $K$, the strategy space $\mathcal{S}_d$ of the defender will then be the set of possible subsets of $\mathcal{T}$ of size $K$. Therefore, each defender will have to choose one of such subsets, denoted by ${s_d} \in \mathcal{S}_d$.

For each defender's choice of a size-$K$ trojans set $s_d \in \mathcal{S}_d$ for which to test and attacker's choice of trojan type $s_a \in \mathcal{S}_a$ to be inserted, the defender's utility function $u_d({s_d},s_a)$ will be:
\begin{equation}
u_d(s_d,s_a)  = \begin{cases}
F_{s_a}  &\textrm{ if } s_a \in s_d,\\
-V_{s_a}&\textrm{ otherwise,}
\end{cases}
\end{equation}
where $V_{s_a}$ is the damage\footnote{$V_{s_a}$ is a mathematical quantization of the volume of the damage that trojan $s_a$ inflicts on the system when activated. Such a quantization requires accurate modeling of the underlying system and the interconnection between its various components. The incorporation of the system model in the problem formulation can be treated in a future work.} done by trojan $s_a$ if it goes undetected. Given the zero-sum nature of the game, the utility of the attacker is simply $u_a(s_d,s_a) = - u_d(s_d,s_a)$.

\section{Prospect Theory for Hardware Trojan Detection: Uncertainty and Risk in Decision Making}\label{sec:pt}
\subsection{Mixed Strategies and Expected Utility Theory}
For the studied hardware trojan detection game,  it is reasonable to assume that both defender and attacker make probabilistic choices over their strategies; and therefore, we are interested in studying the game under \emph{mixed strategies}~\cite{TB00} rather than under \emph{pure, deterministic strategies}. The rationale for such mixed probabilistic choices is two-fold: a) both attacker and defender must randomize between their strategies so as not to make it trivial for the opponent to guess their strategy and b) the hardware trojan detection game can be repeated over an infinite horizon; and therefore, mixed strategies allow capturing the frequencies with which the attacker or defender would use a certain strategy.

To this end, let $\boldsymbol{p}=[\boldsymbol{p}_d\ \boldsymbol{p}_a]$ be the vector of mixed strategies of both players where, for the defender, each element in $\boldsymbol{p}_d$ is the probability with which the defender chooses a certain size-$K$ subset $s_d \in \mathcal{S}_d$ of trojans for which to test; and for the attacker, each element in $\boldsymbol{p}_a$ represents the probability with which the attacker chooses to insert a trojan $s_a \in \mathcal{S}_a$.

In traditional game theory~\cite{TB00}, it is assumed that players act rationally. This rational assumption implies that each player, attacker or defender, will objectively choose its mixed strategy vector so as to optimize its expected utility. Indeed, under conventional expected utility theory (EUT), the utility of each player is simply the expected value over its mixed strategies which, for any of the two players $i\in \mathcal{N}$, is given by:\\%\vspace{-0.15cm}
\begin{align} \label{eq:multiplayerET}%\vspace{-0.15cm}
\begin{split}
&U_i^{\text{EUT}}( \boldsymbol{p}_d,\boldsymbol{p}_a)=\sum_{\boldsymbol{s} \in \mathcal{S}}\bigg(p_d(s_d)p_a(s_a)\bigg) u_i(\boldsymbol{s}),
\end{split}
\end{align}
where $\boldsymbol{s}=[s_d\ s_a]$ is a vector of selected pure strategies and $\mathcal{S}=\mathcal{S}_d \times \mathcal{S}_a$.

\subsection{Prospect Theory for the Hardware Trojan Detection Game}\label{subsec:pt}\vspace{-0.05cm}
In conventional game theory, EUT allows the players to evaluate an objective expected utility such as in (\ref{eq:multiplayerET}) in which they are assumed to act rationally and to objectively assess their outcomes. However, in real-world experiments, it has been observed that users' behavior can deviate considerably from the rational behavior predicted by EUT. The reasons for these deviations are often attributed to the risk and uncertainty that players often face when making decisions over game-theoretic outcomes.

In particular, several empirical studies~\cite{kahneman1979prospect,prelec1998probability,PT01,PT02,PT03,PT04} have demonstrated that when faced with decisions that involve gains and losses under risks and uncertainty, such as in the proposed hardware trojan detection game, players can have a subjective evaluation of their utilities. In the studied game, both the attacker and defender face several uncertainties. In fact, the defender can never be sure of which type of trojans the attacker will be inserting; and thus, when evaluating its outcomes using (\ref{eq:multiplayerET}), it may overweight or underweight the mixed-strategy vector of the attacker $\boldsymbol{p}_a$. Similarly, the attacker may also evaluate its utility given a distorted and uncertain view of the defender's possible strategies. In addition, the decisions of both attacker and defender involve humans (e.g., administrators at the governmental agency or hackers at the manufacturer) who might guide the way in which trojans are inserted or tested for. This human dimension will naturally lead to potentially irrational behavior that can be risk averse or risk seeking; thus, deviating from the rational tenets of classical game theory and EUT.

For the proposed game, such considerations of risk and uncertainty in decision making can translate into the fact that each player $i$ must decide on its action, in the face of the uncertainty induced by the mixed strategies of its opponent, which impacts directly the utility as in (\ref{eq:multiplayerET}). In order to capture such risk and uncertainty factors in the proposed hardware detection game, we turn to the \emph{emerging framework of prospect theory (PT)} \cite{kahneman1979prospect}.

One important notion from PT that is useful for the proposed hardware trojan detection game is the so-called \emph{weighting effect} on the game's outcomes. For instance, PT studies~\cite{kahneman1979prospect,prelec1998probability,PT01,PT02,PT03,PT04} have demonstrated that, in real-life, players of a certain adversarial or competitive game tend to introduce subjective weighting of outcomes that are subject to uncertainty or risk. For the hardware trojan detection game, we use the weighting effect as a way to measure how each player can view a distorted or subjective evaluation of the mixed strategy of its opponents. This subjective evaluation represents the limits on the rationality of the defender and attacker under the uncertainty and lack of exact knowledge of the possible actions of the adversary.

Thus, under PT considerations, for a player $i \in \mathcal{N}$, instead of objectively perceiving the mixed strategy ${{p}_{j}}$ chosen by the adversary, each player views a weighted or distorted version of it, $w_i({p_{j}})$, which is a nonlinear transformation that maps an objective probability to a subjective one. The exact way in which this transformation is defined is based on recent empirical studies in~\cite{kahneman1979prospect,prelec1998probability,PT01,PT02,PT03,PT04}, which show that players, in real-life decision making, tend to underweight high probability outcomes and overweight low probability outcomes\cite{kahneman1979prospect}. For our analysis, for each player $i$, we choose the widely used Prelec function which can capture the previously mentioned weighting effect as follows~\cite{prelec1998probability} (for a given probability $p_i$):\\ \vspace{-0.2cm}
%\vspace{-0.05cm}
\begin{align} \label{eq:weight}%\vspace{-0.15cm}
w_i(p_i)=\exp(-(-\ln p_i)^{\alpha_i}),\ 0<\alpha_i \le 1,
\end{align}
where $\alpha_i$ will be referred to as the \emph{rationality parameter} which allows to express the distortion between player $i$'s subjective and objective probability perception. This parameter allows characterizing how rational the attacker or defender is by measuring how much the uncertainty and risk that this player faces distort its view of the opponents' probability. Note that when $\alpha_i=1$, this is reduced to the conventional EUT probability with full rationality. An illustration on the impact of $\alpha_i$ is shown in Fig.~\ref{fig:alpha}.

The Prelec function has been widely used to model the weighting effect of PT due to its mathematical properties which allow it to fit various experimental observations~\cite{prelec1998probability,kahneman1979prospect}. These properties include: 1) the Prelec function, $w(p)$, is regressive indicating that at the start of the range of definition of $p$, $w(p)>p$, but then afterwards, $w(p)<p$ , 2) the Prelec function has an S-shape which captures the fact that it is first concave then convex, and 3) $w(p)$ is asymmetric with fixed point and inflection point at $p=1/e\approx0.37$. This has made the Prelec function widely used in PT models such as the case in~\cite{PTM00,WSM00}, and~\cite{WalidProspect}, among others. Here, we note that a number of alternative weighting functions have also been derived in literature and are discussed thoroughly in~\cite{OntheShapeofWeighting}. To derive other functions, real-world experiments with real human subjects are needed. However, in general, our proposed framework can accommodate any weighting function.
\begin{figure}[!t]
\begin{center}
\includegraphics[width=8cm]{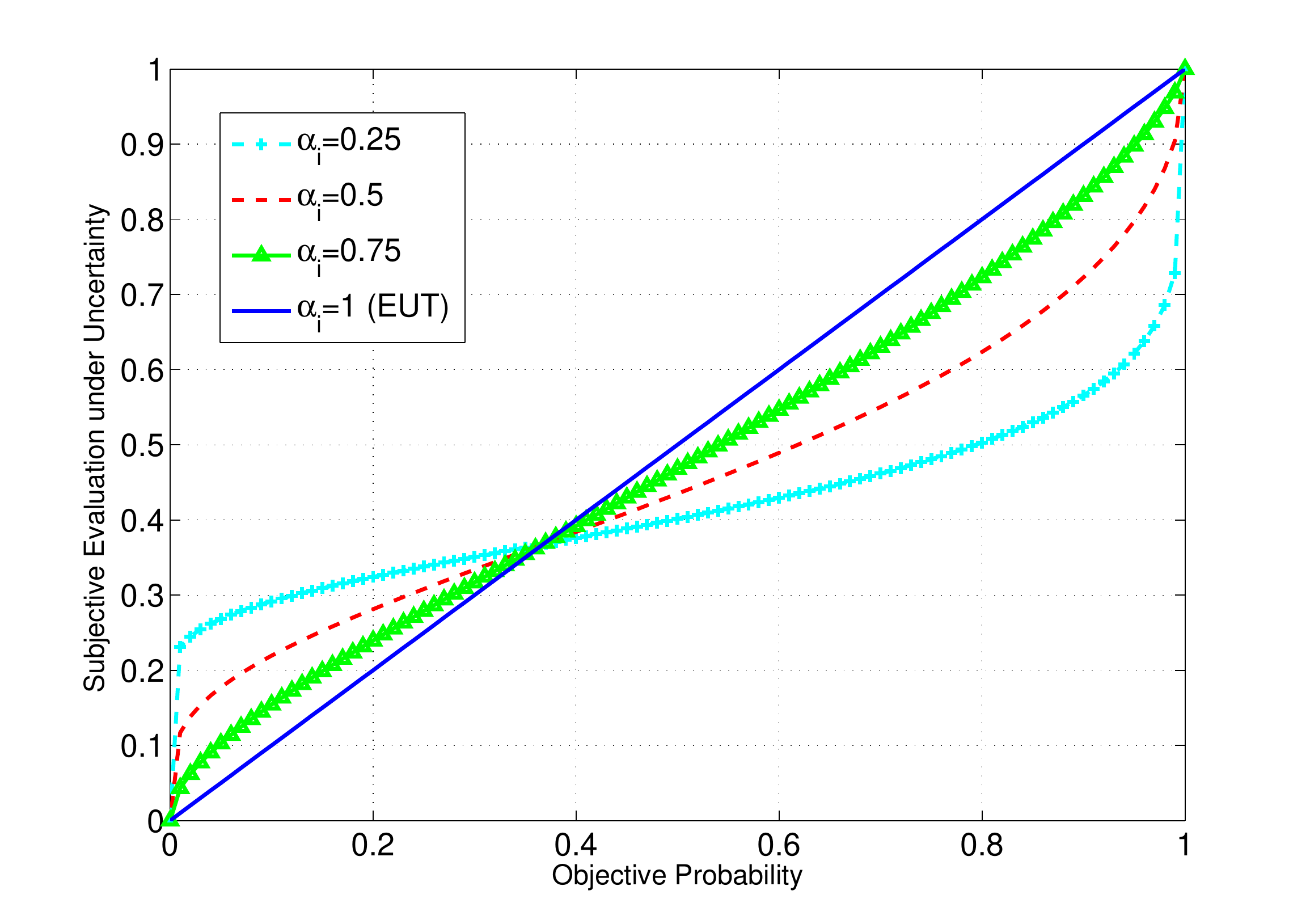}
\end{center}\vspace{-0.4cm}
\caption {Illustration of the impact of the rationality parameter $\alpha_i$.} \label{fig:alpha}\vspace{-0.6cm}
\end{figure}

Given these PT-based uncertainty and risk considerations, the expected utility achieved by a player $i$ will thus be:\vspace{-0.15cm}

\begin{align} \label{eq:multiplayerPT}%\vspace{-0.15cm}
\begin{split}
&U_i^{\text{PT}}( \boldsymbol{p}_i,\boldsymbol{p}_j) = \sum_{\boldsymbol{s} \in \mathcal{S}}\bigg(p_i(s_i) w_{i}(p_{j}(s_{j}))\bigg) u_i(s_i, s_{j}),
\end{split}
\end{align}\vspace{-0.3cm}\\
where $i$ and $j$ correspond, respectively, to the defender and attacker, and vice versa. 
Clearly, in (\ref{eq:multiplayerPT}), the uncertainty is captured via each player's weighting of its opponent's strategy. This weighting depends on the rationality of the player under uncertainty, which can be captured by $\alpha_i$.

Given this re-definition of the game, our next step is to study and discuss the game solution under both EUT and PT.

\section{Game Solution and Proposed Algorithm}\label{sec:algo}
\subsection{Mixed-Strategy Nash Equilibrium}
To solve the proposed game, under both EUT and PT, we seek to characterize the \emph{mixed-strategy Nash equilibrium} of the game:
\begin{definition}\label{def:MSNE}
A mixed strategy profile $\boldsymbol{p}^*$ is said to be a mixed strategy Nash equilibrium if for the defender, $d$, and attacker, $a$, we have:
\begin{align}\label{eq:ne}
U_d(\boldsymbol{p}_d^*,\boldsymbol{p}_{a}^*) \ge U_d(\boldsymbol{p}_d,\boldsymbol{p}_{a}^*), \  \forall p_d \in  \mathcal{P}_d, \nonumber\\
U_a(\boldsymbol{p}_a^*,\boldsymbol{p}_{d}^*) \ge U_a(\boldsymbol{p}_a,\boldsymbol{p}_{d}^*), \  \forall p_a \in  \mathcal{P}_a,
\end{align}
\end{definition}
\noindent
where $\mathcal{P}_i$ is the set of all probability distributions available to player $i$ over its action space $\mathcal{S}_i$. Note that, the mixed-strategy Nash equilibrium definition in (\ref{eq:ne}) is applicable for both EUT or PT, the difference would be in whether one is using (\ref{eq:multiplayerET}) or (\ref{eq:multiplayerPT}), respectively.

The mixed-strategy Nash equilibrium (MSNE) represents a state of the game in which neither the defender nor the attacker has an incentive to unilaterally deviate from its current mixed-strategy choice, given that the opposing player uses an MSNE strategy. Under EUT, this implies that under a rational choice the MSNE represents the case in which the defender has chosen its optimal randomization over its testing strategies and, thus, cannot improve its utility by changing these testing strategies; assuming that the attacker is also rational and utility maximizing as per EUT. Similarly, for the attacker, an MSNE under EUT implies that the attacker has chosen its optimal randomization over its choice of trojan to insert and, thus, cannot improve its utility by changing this choice of trojan; assuming that the defender is also rational and utility maximizing as per EUT.
Under PT, at the MSNE neither the attacker nor the defender can improve their perceived and subjective utility evaluation as per (\ref{eq:multiplayerPT}) by changing their MSNE strategies given their rationality levels captured by $\alpha_d$ and $\alpha_a$. Thus, under PT the MSNE is a state of the game in which neither the defender nor the attacker can further improve their utilities by unilaterally deviating from the MSNE, under their current uncertain perception on one another.

Given the zero-sum two-player nature of the game, finding closed-form solutions for the MSNE can follow the von Neumann indifference principle~\cite{TB00} under which, for each player at the MSNE, the expected utilities of any pure strategy choice, under the mixed strategies played by the opponent, are equal. Such a principle can be trivially shown to be applicable to both EUT and PT due to the one-to-one relationship between the probabilities and the weights. For the proposed game, given the large strategy space of both defender and attacker, it is challenging to solve the equations that stem from the indifference principle, for a general case. However, as will be shown for a numerical case study in Section~\ref{sec:num}, the game may admit multiple equilibrium points. Therefore, given an initial starting point of the system, one must develop learning algorithms~\cite{LEARN07} to characterize one of the MSNEs, as proposed next.
\subsection{Proposed Algorithm: Fictitious Play and Convergence Results}
To solve the studied hardware trojan detection game, under both EUT and PT, we propose a learning algorithm, summarized in Table~\ref{alg:alg1}, which is based on the fictitious play (FP) algorithm~\cite{TB00,LEARN07}. In this algorithm, each player uses its belief about the mixed strategy that its opponent will adopt. This belief stems from previous observations and is updated in every iteration.
\begin{algorithm}[!t]
\caption{Distributed Fictitious Play Learning Algorithm}
\label{alg:alg1}
\begin{algorithmic}[1]
\REQUIRE Action space of the defender, $\mathcal{S}_d$
\\\setlength\parindent{15pt} Action space of the attacker, $\mathcal{S}_a$
\\\setlength\parindent{15pt} Convergence parameter, $M$
\ENSURE Equilibrium mixed strategy vector of each player, $\boldsymbol{p}_d^*$ and $\boldsymbol{p}_a^*$ 
\STATE Initialize $\boldsymbol{\sigma}^0_a$ and $\boldsymbol{\sigma}^0_d$
\STATE Initialize convergence tester: $C_{\textrm{test}}=0$
\STATE Initialize iteration counter: $k=1$
\WHILE {Not Converged: $C_{\textrm{test}} ==0$}
\STATE Each player chooses its optimal strategy: 
\\ $s_d^k=\arg\max_{s_d\in\mathcal{S}_d}{U}_{d}(s_d,\boldsymbol{\sigma}_{d}^{k-1})$
\\ $s_a^k=\arg\max_{s_a\in\mathcal{S}_a}{U}_{a}(s_a,\boldsymbol{\sigma}_{a}^{k-1})$
\STATE Each player updates its observed empirical frequency:
\\ $\sigma_d^{k}(s_a)=\frac{k-1}{k} \cdot \sigma_d^{k-1}(s_a) + \frac{1}{k}\cdot \mathbf{1}_{\{s_a^{k-1}=s_a^{k}\}} \forall s_a\in\mathcal{S}_a ,$
\\ $\sigma_a^{k}(s_d)=\frac{k-1}{k} \cdot \sigma_a^{k-1}(s_d) + \frac{1}{k}\cdot \mathbf{1}_{\{s_d^{k-1}=s_d^{k}\}}, \forall s_d\in\mathcal{S}_d$
\STATE Check Convergence
\\ Calculate: 
\\ $C_d(s_a)=|\sigma_d^k(s_a)-\sigma_d^{k-1}(s_a)| \,\forall s_a\in\mathcal{S}_a ,$
\\ $C_a(s_d)=|\sigma_a^k(s_d)-\sigma_a^{k-1}(s_d)| \,\forall s_d\in\mathcal{S}_d $
\IF {Converged: \\$C_d(s_a)<\frac{1}{M} \,\,\&\& \,\,C_a(s_d)<\frac{1}{M}$ $\forall s_a\in\mathcal{S}_a$, $\forall s_d\in\mathcal{S}_d$ }
\STATE $C_{\textrm{test}}=1$
\STATE Compute strategy vectors:
\\ $\boldsymbol{p}_d^*=[\sigma_a^k(s_1), \sigma_a^k(s_2), ..., \sigma_a^k(s_{|\mathcal{S}_d|})]$
\\ $\boldsymbol{p}_a^*=[\sigma_d^k(s_1), \sigma_d^k(s_2), ..., \sigma_d^k(s_{|\mathcal{S}_d|})]$
\ENDIF
\STATE Update Counter: k=k+1
\ENDWHILE
\RETURN Strategy vectors: $\boldsymbol{p}_d^*$ and $\boldsymbol{p}_a^*$  
\end{algorithmic}
\end{algorithm}
In this regard, let $\boldsymbol{\sigma}_i^k$ be player $i$'s perception of the mixed strategy that $j$ adopts at time instant $k$. Here, each entry of $\boldsymbol{\sigma}_i^k$, given by $\sigma_i^k(s_j)$, represents the belief that $i$ has at time $k$ of the probability with which $j$ will play the strategy $s_j\in\mathcal{S}_j$. Such perception can be built based on the empirical frequency with which $j$ has used $s_j$ in the past. 
Thus, let $\eta_i^k(s_j)$ be the number of times that $i$ has observed $j$ playing strategy $s_j$ in the past, up to time instant $k$. Then, $\sigma_i^k(s_j)$ for each $s_j\in\mathcal{S}_j$ can be calculated as follows:
\begin{align}\label{eq:Probbelief}
\sigma_i^k(s_j)=\frac{\eta_i^k(s_j)}{\sum_{s'_j\in\mathcal{S}_j}\eta_i^k(s'_j)}.
\end{align}\vspace{0.01cm}

To this end, at time instant $k+1$, based on the vector of empirical probabilities that it has perceived until time instant $k$, $\boldsymbol{\sigma}_{i}^{k}$, each player $i$ chooses the strategy $s_i^{k+1}$ that maximizes its expected utility:
\begin{align}\label{eq:MaxUtility}
s_i^{k+1}=\arg\max_{s_i\in\mathcal{S}_i}{U}_{i}(s_i,\boldsymbol{\sigma}_{i}^{k}),
\end{align}
where the expected utility is calculated as the expected value of the the utility achieved by player $i$, when choosing strategy $s_i$, with respect to the perceived probability distribution at time instant $k$ over the set of actions of the opponent, $\boldsymbol{\sigma}_{i}^{k}$. This is equivalent to the notion of expected utility that we derived in~(\ref{eq:multiplayerET}) and~(\ref{eq:multiplayerPT}) under, respectively, EUT and PT.
 
After each player $i$ chooses its strategy at time instant $k+1$, it can update its beliefs as follows:
\begin{align}\label{eq:UpdateFreq}
\sigma_i^{k+1}(s_j)=\frac{k}{k+1} \cdot \sigma_i^k(s_j) + \frac{1}{k+1}\cdot \mathbf{1}_{\{s_j^k=s_j^{k+1}\}},
\end{align}
which is equivalent to calculating $\sigma_i^{k+1}(s_j)$ based on~(\ref{eq:Probbelief}).

In summary, at iteration $k+1$, player $i$ observes the actions of its opponent up to time $k$ and updates its perception of its opponent's mixed strategy based on~(\ref{eq:Probbelief}) or, equivalently,~(\ref{eq:UpdateFreq}). 
Subsequently, at time $k+1$, player $i$ chooses a strategy $s_i^{k+1}$ from its available strategy set $\mathcal{S}_i$ which maximizes its expected utility with respect to its updated perceived empirical frequencies as shown in~(\ref{eq:MaxUtility}). This expected utility would follow (\ref{eq:multiplayerET}) for EUT and (\ref{eq:multiplayerPT}) for PT. However, in the case of PT, after computing the empirical frequency based on~(\ref{eq:Probbelief}), these frequencies are weighed based on~(\ref{eq:weight}) such that, when choosing its optimal strategy $s_i^{k+1}$ as in~(\ref{eq:MaxUtility}), each player $i$ uses $\omega_i(\boldsymbol{\sigma}_i^k)$ instead of $\boldsymbol{\sigma}_i^k$. When this weighting is performed, we denote $U_i(s_i,\omega_i(\boldsymbol{\sigma}_i^k))$ by $U_i^{PT}(s_i,\boldsymbol{\sigma}_i^k)$. 

This learning process proceeds until the calculated empirical frequencies converge. Convergence is achieved when:
\begin{align}\label{eq:ConvergenceIt}
|\sigma^{k+1}_i(s_j)-\sigma^{k}_i(s_j)|<\frac{1}{M}, \,\forall s_j\in\mathcal{S}_j,\, \forall i\in\mathcal{N}  
\end{align} 
where $M$ is an arbitrary large number (that typically goes to infinity).

This algorithm requires initialization of the vectors of beliefs. Thus, we let $\boldsymbol{\sigma}_d^{0}$ and $\boldsymbol{\sigma}_a^{0}$ be the initial values adopted, respectively, by the defender and attacker. Such initialization vectors can be based on previous experience or can be any arbitrary probability distribution over the action space of the opponent.
This algorithm is shown in details in Table~\ref{alg:alg1}.  

For a two-player zero-sum game, it is well known that FP is guaranteed to converge to an MSNE~\cite{TB00,LEARN07}. In other words, it is guaranteed that the empirical frequency that player $i$ builds of the actions of its opponent $j$ converges to, $\boldsymbol{\sigma}_i^*$, which is nothing but the MSNE strategy of its opponent, i.e. $\boldsymbol{p}_j^*$ (defined in Definition~\ref{def:MSNE}). 
The convergence to $\boldsymbol{\sigma}_i^*$ is mathematically defined as the existence of an iteration number $\kappa_i$ such that, for $k>\kappa_i$
the belief of player $i\in\mathcal{N}$, $\boldsymbol{\sigma}_i^k$, converges to $\boldsymbol{\sigma}^*_i$, i.e. 
\begin{align}\label{eq:Convergence}
|\sigma^k_i(s_j)-\sigma_i^*(s_j)|<\frac{\epsilon}{M} \forall s_j\in\mathcal{S}_j,
\end{align} 
where $M$ is an arbitrary large number (that typically goes to infinity) and $\epsilon$ is a positive constant.
Hence, for our studied case, $\boldsymbol{\sigma}_a^*$ converges to to the MSNE of the defender, $\boldsymbol{p}_d^*$, and $\boldsymbol{\sigma}_d^*$ converges to the MSNE of the attacker, $\boldsymbol{p}_a^*$. However, to our knowledge, such a result has not been extended to PT, as done in the following theorem: 
%%%%%%%%%%%----
\begin{theorem}\label{th:cov2}
For the proposed hardware trojan detection game, the proposed FP-based algorithm is guaranteed to converge to a mixed NE under both EUT and PT.
\end{theorem}
\begin{proof}
The convergence of FP to an MSNE for EUT in a two-player zero-sum game is a known result~\cite{LEARN07,monderer1996fictitious,TB00}. For PT, one can easily verify that the convergence to a fixed point will directly follow from the EUT results in~\cite{LEARN07,monderer1996fictitious,TB00}. However, what remains to be shown is that this convergence will actually reach an MSNE for the case of PT. We prove this case using contradiction as follows.

Suppose that $\{\boldsymbol{\sigma}^k\}$ is a fictitious play process that will converge to a fixed point and a mixed strategy $\boldsymbol{p}^*$ after $k$ iterations (i.e. $\boldsymbol{\sigma}_i^k$ converges to $\boldsymbol{p}_j^*$ for both players after $k$ iterations). If the vector $\boldsymbol{p}^*=\{\boldsymbol{p}^*_i,\boldsymbol{p}^*_{j}\}$ is not an MSNE, then there must exist $s_i,\,s'_i \in \mathcal{S}_i$, such that $p_i^*(s_i)>0$ and
$$
U_i^{\text{PT}}\biggl(s'_i,\boldsymbol{p}^*_{j}\biggr)>U_i^{\text{PT}}\biggl(s_i,\boldsymbol{p}^*_{j}\biggr),
$$
where $U_i^{\text{PT}}(s'_i,\boldsymbol{p}^*_{j})$ is the expected utility with respect to the mixed strategies of $j$, the opponents of player $i$, when player $i$ chooses pure strategy $s'_i$. Here, we can choose a value $\epsilon$ that satisfies 
\begin{align}\label{eq:epsilon}
0<\epsilon<\frac{1}{2}|U_i^{\text{PT}}(s'_i,\boldsymbol{p}^*_{j})-U_i^{\text{PT}}(s_i,\boldsymbol{p}^*_{j})| 
\end{align} 
as $\boldsymbol{\sigma}^k$ converges to $\boldsymbol{p}^*$ at iteration $k$. Also, since the FP process decreases as the number of iterations $n$ increases, the utility distance of a pure strategy between two neighboring iterations must be less than $\epsilon$ after a certain iteration $k$. For $n \ge k$, the FP process can be written as:\\
\begin{align}\label{eq:ineq}
\begin{split}
U_i^{\text{PT}}(s_i,\boldsymbol{\sigma}_{i}^n)=&\sum_{s_j \in \mathcal{S}_j} u_i(s_i,s_{j}) w_i(\sigma_{i}^n(s_j))\\
\le &\sum_{s_j \in \mathcal{S}_j} u_i(s_i,s_j) w_i(p^*_j(s_j))+\epsilon\\
< &\sum_{s_j \in \mathcal{S}_j} u_i(s'_i,s_j) w_i(p_j^*(s_j))-\epsilon\\
\le &\sum_{s_j \in \mathcal{S}_j} u_i(s'_i,s_{j}) w_i(\sigma_i^n(s_j))\\
=&U_i^{\text{PT}}(s'_i,\boldsymbol{\sigma}_{i}^n),
\end{split}
\end{align}\\
where the two equalities in~(\ref{eq:ineq}) stem directly from the definition of expected utility given in~(\ref{eq:multiplayerPT}) when $i$'s strategy is fixed to $s_i$ or $s'_i$, and the transition from step 2 to step 3 stem directly from~(\ref{eq:epsilon}). Thus, player $i$ would not choose $s_i$ but would rather choose $s'_i$ after the $n^\textrm{th}$ iteration; mathematically, we will have $p_i(s_i)=0$ and $w_{j}(\sigma_j(s_i))=0$ (the other player's perception of $p_i(s_i)$). Hence, we get $p_i(s_i)=0$ which contradicts the initial assumption that $p_i(s_i)>0$; thus the theorem is shown.
\end{proof}
%Following the convergence to a mixed-strategy Nash equilibrium, the last stage in the Algorithm of Table~\ref{alg:alg1} is the actual negotiation phase which can occur periodically between the defender and attacker. In this stage, after a period of time during which testing and hardware insertion occurs, the defender assesses its damages and the attacker assesses any possible fines and, subsequently, the two may directly interact or negotiate future trade or exchange based on past results. For example, after a certain period of time has passed, the defender might find that a certain manufacturer has been highly malicious and then it decides to severe all ties to it and/or replace it with another. The actual process of Stage 3 is beyond the scope of this paper and will follow economic and real-world contract negotiations that could be interesting to study in future work.

\section{Numerical Case Study: Analytical and Simulation Results}\label{sec:num}
For simulating the hardware trojan detection game, we consider the scenario in which the attacker, denoted hereinafter by player $1$, has four types of trojans (strategies) $A$, $B$, $C$, and $D$, i.e., $\mathcal{S}_a=\mathcal{T}=\{A,B,C,D\}$ whose damage values are $V_A=1$, $V_B=2$, $V_C=4$, and $V_D=12$. These numbers are used to illustrate different damage levels to the system. For example, these values can be viewed as monetary losses to the defender and, hence, attacking gains to the attacker. Given that there are no existing empirical data on the hardware detection game, we have chosen illustrative numbers that show four varying levels of damage. However, naturally, the subsequent analysis may be extended to analyze the game under other damage values. In this scenario, we assume that the defender, referred to as player $2$, can test for $K=2$ types of trojans at a time and, thus, it has $6$ strategies. Without loss of generality, we assume that the fine is similar for all types of trojans, i.e., $F_{s_a} = F\ \forall s_a\in \mathcal{S}_a$.

For this numerical case study, we will first derive several analytical results that allow us to gain more insights on the proposed hardware detection game under both EUT and PT representations. Then, we present several simulation results that provide additional insights and analysis on the proposed game and on the impact of PT consideration in the game model.

\subsection{Analytical Results}\label{sec:analy}
In this subsection, we derive a series of results to gain more insights on the Nash equilibria of the game as well as on the possible values of the fine and how they impact the game under both PT and EUT. First, we can state the following theorem with regard to the Nash equilibria of the game under both EUT and PT:

\begin{theorem}\label{Theorem:AttackerUniqueness}
When $F>0$, under both EUT and PT, the proposed game can admit multiple equilibria. However, in all of these equilibria, the attacker has the same mixed-strategy Nash equilibrium strategies.
\end{theorem}
\begin{proof}
To capture the pure strategy payoffs of the defender given the attacker's mixed strategy, we use the indifference principle as per the following equation:
\begin{equation}\label{eq:Ma}
\begin{split}
\boldsymbol{U}_d(\boldsymbol{s}_d,\boldsymbol{p}_a^*)=&\boldsymbol{M}_a \cdot \boldsymbol{p}_a^*\\
\begin{bmatrix}
U_d(AB,\boldsymbol{p}_a^*)\\
U_d(AC,\boldsymbol{p}_a^*)\\
U_d(AD,\boldsymbol{p}_a^*)\\
U_d(BC,\boldsymbol{p}_a^*)\\
U_d(BD,\boldsymbol{p}_a^*)\\
U_d(CD,\boldsymbol{p}_a^*)\\
\end{bmatrix}
=&
\begin{bmatrix}
F&F&-4&-12\\
F&-2&F&-12\\
F&-2&-4&F\\
-1&F&F&-12\\
-1&F&-4&F\\
-1&-2&F&F\\
\end{bmatrix}\cdot
\begin{bmatrix}
p_a^*(A)\\
p_a^*(B)\\
p_a^*(C)\\
p_a^*(D)\\
\end{bmatrix},\\
\end{split}
\end{equation}
where $p_a^*(A)+p_a^*(B)+p_a^*(C)+p_a^*(D)=1$.  $\boldsymbol{M}_a$ is the utility matrix of the attacker which we can use to obtain the attacker's $\boldsymbol{p}_a^*$. Using the indifference principle, for the defender, an  MSNE must satisfy $U_d(AB,\boldsymbol{p}_a^*)=U_d(AC,\boldsymbol{p}_a^*)=U_d(AD,\boldsymbol{p}_a^*)=U_d(BC,\boldsymbol{p}_a^*)=U_d(BD,\boldsymbol{p}_a^*)=U_d(CD,\boldsymbol{p}_a^*).$ 
In addition,
\begin{align}
\text{rank}(\boldsymbol{M}_a)=\text{rank}
\left(\begin{bmatrix}
F&F&-4&-12\\
0&-2-F&F+4&0\\
0&-2-F&0&F+12\\
-1&F&F&-12\\
0&0&0&0\\
0&0&0&0\\
\end{bmatrix}\right)=4.
\end{align}
Thus, the attacker has only one solution $\boldsymbol{p}_a^*$ since: {1)} the rank of the attacker's utility matrix is equal to the dimension of its mixed strategy and {2)} the auxiliary equation $\sum_{s_a} p_a^*(s_a)=1$ balances the requirement of $\boldsymbol{U}_d(\boldsymbol{s}_d,\boldsymbol{p}^*_a)$ in (\ref{eq:Ma}). Similarly, we capture the pure strategy payoffs of the attacker via the defender's mixed strategy:
\begin{align}\label{eq:Md}
\boldsymbol{U}_a(\boldsymbol{s}_a,\boldsymbol{p}_d^*)=\boldsymbol{M}_d \cdot \boldsymbol{p}_d^*,
\end{align}
where $\boldsymbol{M}_d=-\boldsymbol{M}_a^{T}$. In particular, $\sum_{s_d} p_d^*(s_d)=1$ and $U_a(A,\boldsymbol{p}_d^*)=U_a(B,\boldsymbol{p}_d^*)=U_a(C,\boldsymbol{p}_d^*)=U_a(D,\boldsymbol{p}_d^*)$ at the MSNE for the attacker. Since the rank of the defender's utility matrix, $\text{rank}(\boldsymbol{M}_d)=\text{rank}(\boldsymbol{M}_a)=4$, is less than the number of defender's strategies, we get multiple defender MSNE strategies.

As an example, when $F=8$ in (\ref{eq:Ma}) and (\ref{eq:Md}), we could obtain the only NE for the attacker under EUT, $\boldsymbol{p}_a^*=[0.32, 0.29, 0.24, 0.16]^T$. Also, we can compute the multiple NEs of the defender under EUT:
\begin{equation}\label{eq:p2solu}
\begin{split}
p_d^*(AB)=&-0.2259 + p_d(CD), \\
p_d^*(AC)=& -0.1290 + p_d(BD), \\
p_d^*(AD) =&0.7097 - p_d(BD) - p_d(CD), \\
p_d^*(BC)=& 0.6452- p_d(BD) -p_d(CD).
\end{split}
\end{equation}

Under PT, the auxiliary equation is equivalent to $\sum_{p} \exp(-(-\ln w(p))^{\frac{1}{\alpha}})=1$. This equation does not change the ranks of neither $\boldsymbol{M}_a$ nor $\boldsymbol{M}_d$ nor the number of eigenvalues. Thus, based on Cayley-Hamilton theorem, PT and EUT have the same number of eigenvalues and then, have the same number of MSNEs. This is applicable for any value of the fine. 
\end{proof}

Next, we show that, for both EUT and PT, there exists a value $F^v$ for the fine at which neither the attacker nor the defender will win, i.e., the value of the game is zero:

\begin{theorem}\label{th:FEUT}
For EUT and PT, at the MSNE, there exists a fine value, respectively, $F^v_{\textrm{EUT}}$ and $F^v_{\textrm{PT}}$ such that neither the attacker nor the defender wins.
\end{theorem}
\begin{proof}
See the Appendix.
\end{proof}

Given Theorem~\ref{th:FEUT}, we can show the following result:
\begin{corollary}
There exists a minimum fine value $F_{\textrm{min}}$, such that, the utility of the attacker will be positive (the attacker wins overall) under both EUT and PT, i.e. $U_a>0, U_d<0$.
\end{corollary}
\begin{proof}
Based on Theorem~\ref{th:FEUT}, it can be shown that the utilities of the attacker and defender intersect at $0$. Also, the derivative of the utility with respect to $F$ can be easily seen to be monotonic. Thus, there exists a fine $F_{\textrm{min}}$ such that $U_a>0$ and $U_d<0$.
\end{proof}
\vspace{-0.2cm}
\begin{remark}
The generalization of the results in Theorems~\ref{Theorem:AttackerUniqueness}~and~\ref{th:FEUT} is directly dependent on the general computation of the rank of matrix $\boldsymbol{M}_a$ for an arbitrary number of trojan types, number of types for which the defender can simultaneously test, as well as the fine and damage values associated with every trojan type. The derivation of Theorems~\ref{Theorem:AttackerUniqueness}~and~\ref{th:FEUT}, in this section, and the proposed algorithm in Table~\ref{alg:alg1} provide a general methodology which can be followed to derive, respectively, analytical and numerical results for any general trojan detection game. 
\end{remark}

\subsection{Numerical Results}
In this subsection, we run extensive simulations to understand the way PT and EUT considerations impact the hardware trojan detection game. To obtain the mixed Nash equilibrium under EUT and PT, we use the proposed algorithm in Table~\ref{alg:alg1}. The initial strategies are chosen as follows: we choose the attacker's initial strategy set as $p_a=[0.2083\ 0.1667\ 0.3333\ 0.2917]^T$ and the defender's initial strategy set as $p_d=[0.2051\ 0.2564\ 0.2564\ 0.0513\ 0.0513\ 0.1795]^T$. In the subsequent simulations, we assume that the fine for all trojans is equal to $F_{s_a} = F =8,\ \forall s_a \in \mathcal{S}_a$, unless stated otherwise. We vary the values of the rationality parameters $\alpha_a$ (for the attacker) and $\alpha_d$ (for the defender).
\begin{figure}[!t]
 \begin{center}
 \vspace{-0.3cm}
  \includegraphics[width=8cm]{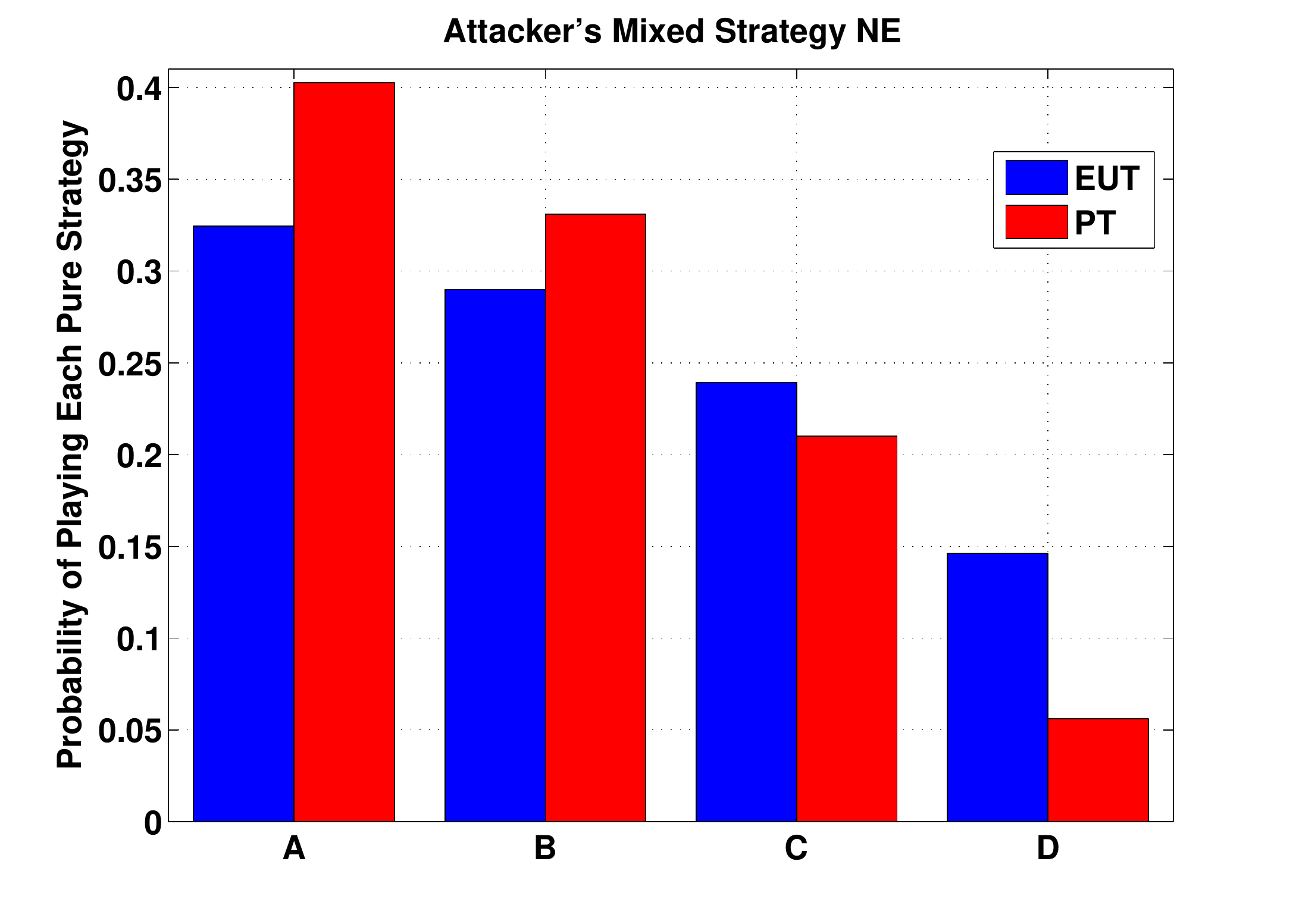}
 \vspace{-0.3cm}
   \caption{\label{fig:atk5} Attacker mixed-strategies at the equilibrium for both EUT and PT with $\alpha_a=\alpha_d=0.5$.}
\end{center}\vspace{-0.3cm}
\end{figure}
\begin{figure}[!t]
 \begin{center}
 \vspace{-0.3cm}
  \includegraphics[width=8cm]{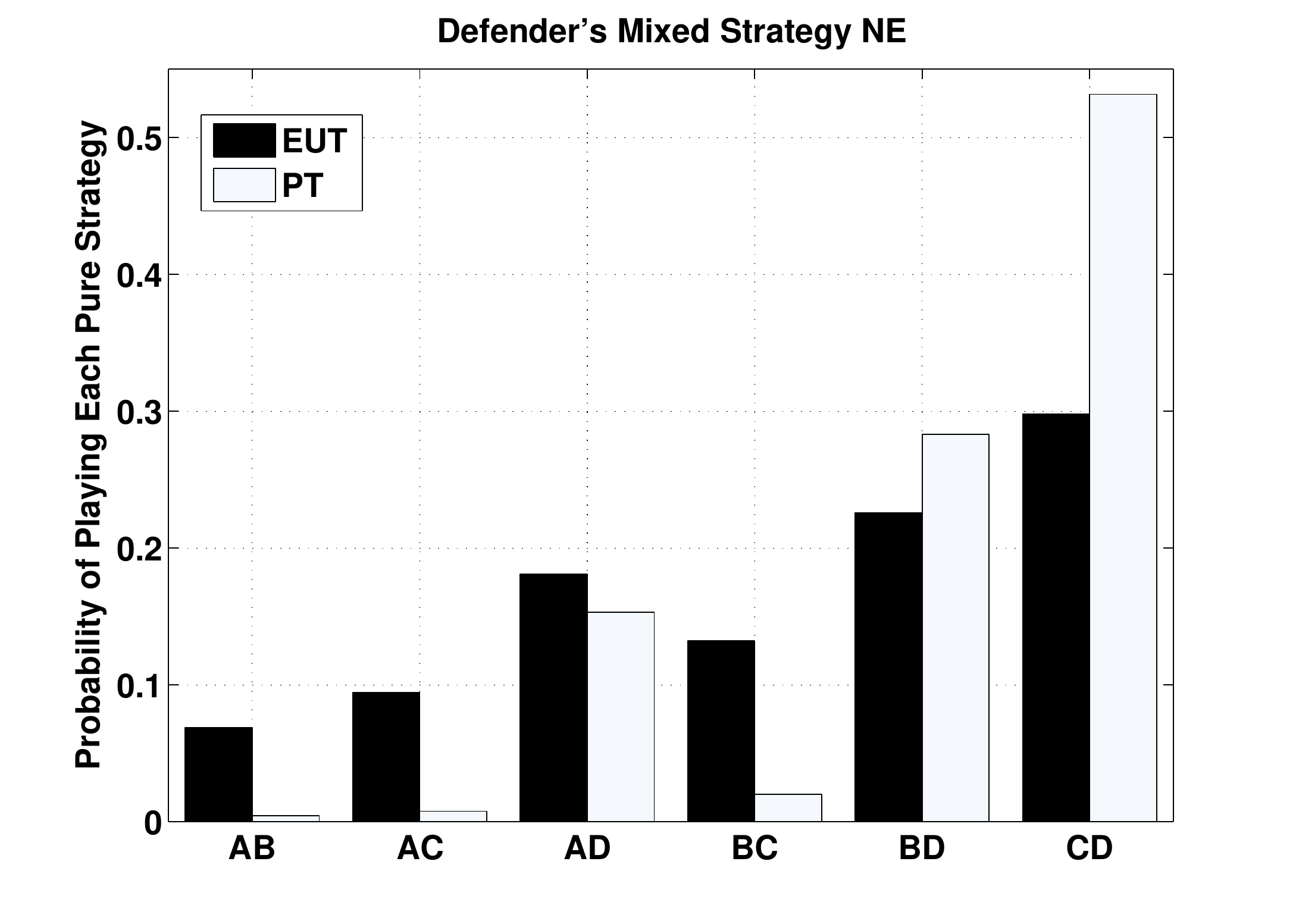}
 \vspace{-0.3cm}
   \caption{\label{fig:dfd5} Defender mixed-strategies at the equilibrium for both EUT and PT with $\alpha_a=\alpha_d=0.5$.}
\end{center}\vspace{-0.6cm}
\end{figure}

Fig.~\ref{fig:atk5} shows the four mixed strategies for the attacker at both the EUT and PT equilibria reached via fictitious play. In this figure as well as in Figs.~\ref{fig:dfd5} and~\ref{fig:ut5}, we choose $\alpha_a=\alpha_d=0.5$ for both attacker and defender under PT reflecting the same level of subjectivity in the behavior of the attacker and defender. Here, we can first see that the equilibrium mixed strategies of the attacker are different between PT and EUT. Under PT, the attacker is more likely to insert trojans such as $A$ or $B$, as compared to EUT, whose value is less than $C$ and $D$. This shows that the attacker becomes more risk averse under PT and, thus, aims at inserting low-valued trojans, rather than focusing on higher valued trojans which are more likely to be detected due to their prospective damage. The impact of such risk aversion on the defender's behavior at the equilibrium is more pronounced as seen in Fig.~\ref{fig:dfd5}. Under PT, the defender will more aggressively attempt to test for the trojan with the highest damage. In this respect, we can see that, under PT, the defender will have a $55\%$ likelihood to test for the two most damaging trojans while ignoring the tests that pertain to trojans $A$ and $B$.

A conservative PT-based defense approach coupled with a risk-averse attacker will naturally lead to a lower overall detection probability and, thus, will lead to further damage to the system, when compared with the fully rational path of EUT. In other words, compared to rational EUT, the attacker is more likely to win in the PT scenario in which both the attacker and the defender deviate from the rational behavior. This result is corroborated in Fig.~\ref{fig:ut5}. In this figure, we show the expected utility of the attacker and defender, at both the EUT and PT equilibria. Clearly, under PT, the attacker is able to incur more damage as compared to EUT, and thus, the overall value of the game decreases from $2.1930$ to $1.5356$; a $30\%$ decrease in utility!

\begin{figure}[!t]
 \begin{center}
 \vspace{-0.3cm}
  \includegraphics[width=8cm]{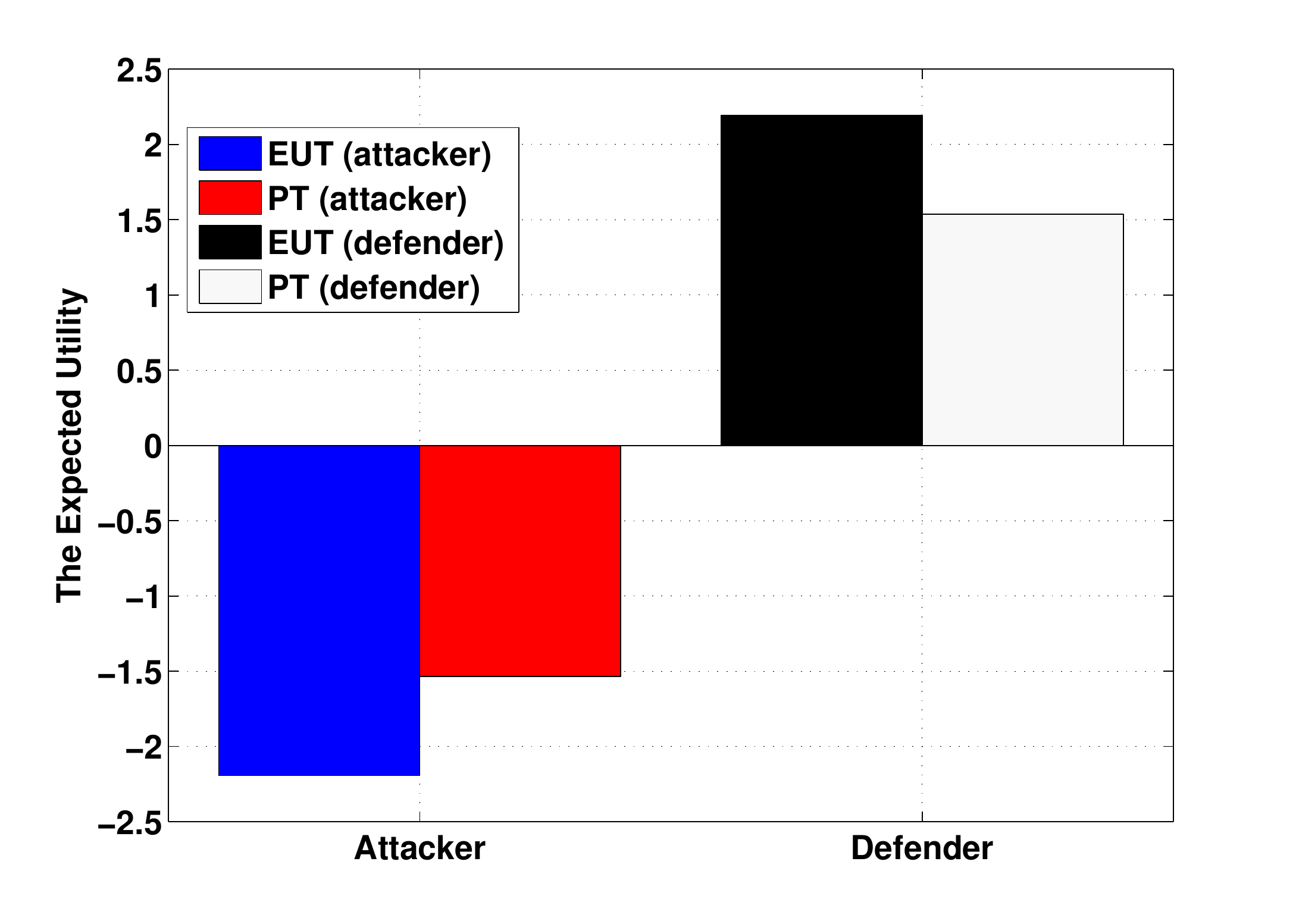}
 \vspace{-0.3cm}
   \caption{\label{fig:ut5} Expected utility at the equilibrium for the attacker and the defender under both EUT and PT with $\alpha_a=\alpha_d=0.5$.}
\end{center}\vspace{-0.3cm}
\end{figure}

In Fig.~\ref{fig:utVF}, we show the expected utility for both PT and EUT, as the value of the fine varies for $\alpha_a=\alpha_d=0.5$. The results in Fig.~\ref{fig:utVF} are used to highlight the impact of the value of the fine and corroborate some of the insights of the theorems in Subsection~\ref{sec:analy}. First, Fig. ~\ref{fig:utVF} shows the expected result that, as the fine value increases, the overall utility achieved by the defender increases while that of the attacker decreases, for both EUT and PT. In this figure, we can see that, under EUT and based on Theorem~\ref{th:FEUT}, the value of the fine for which neither the attacker nor the defender wins is $3$. In particular, at the crossing point, the attacker's MSNE is $\boldsymbol{p}_a^*=[0.3818\ 0.3022\ 0.2133\ 0.1027]$ and $F^v_{\textrm{EUT}}=3.0491$ as in (\ref{eq:FEUT}). For higher values, Fig.~\ref{fig:utVF} shows that under EUT the defender starts achieving a winning utility. More interestingly, we can see through Fig. ~\ref{fig:utVF} that, for PT, the value of the fine for which the utilities are 0 is $4$ which is greater than that of EUT. This implies that, under irrational behavior and uncertainty, the defender must set higher fines in order to start gaining over the attacker. Moreover, Fig.~\ref{fig:utVF} shows that, for this choice of $\alpha_a$ and $\alpha_d$, the defender is better off under EUT rather than PT. 
\begin{figure}[!t]
 \begin{center}
 \vspace{-0.3cm}
  \includegraphics[width=8cm]{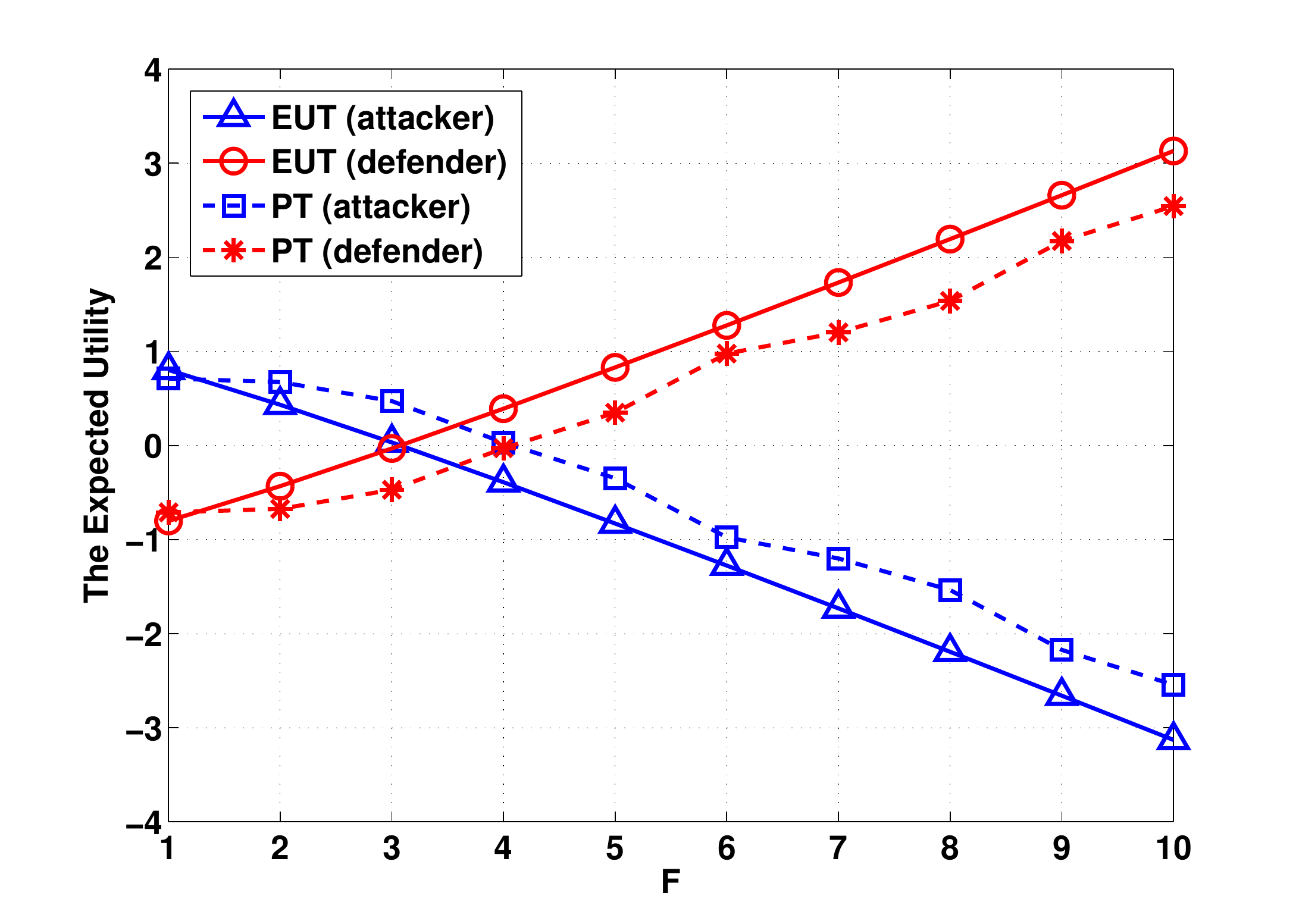}
 \vspace{-0.3cm}
   \caption{\label{fig:utVF} The utility performance as the value of the fine $F$ varies for both EUT and PT with $\alpha_a=\alpha_d=0.5$.}
\end{center}\vspace{-0.6cm}
\end{figure}

Figs.~\ref{fig:Fig6and9}-\ref{fig:ut15} show the equilibrium strategies and corresponding utilities for both EUT and PT, for two scenarios: i) Scenario 1 in which the attacker uses $\alpha_a=0.5$ while the defender is significantly deviating from the rational path, i.e., $\alpha_d=0.1$, and ii) Scenario 2 in which the defender uses $\alpha_d=0.5$ while the attacker is significantly deviating from the rational path, i.e., $\alpha_a=0.1$. 
In this respect, Fig.~\ref{fig:Fig6and9} and Fig.~\ref{fig:Fig7and10} show the mixed strategies of, respectively, the attacker and defender under the two scenarios. In addition, Fig.~\ref{fig:ut51} and Fig.~\ref{fig:ut15} show the expected utility achieved by, respectively, the attacker and defender under the two scenarios.

\begin{figure}[!t]
 \begin{center}
 \vspace{-0.3cm}
  \includegraphics[width=8cm]{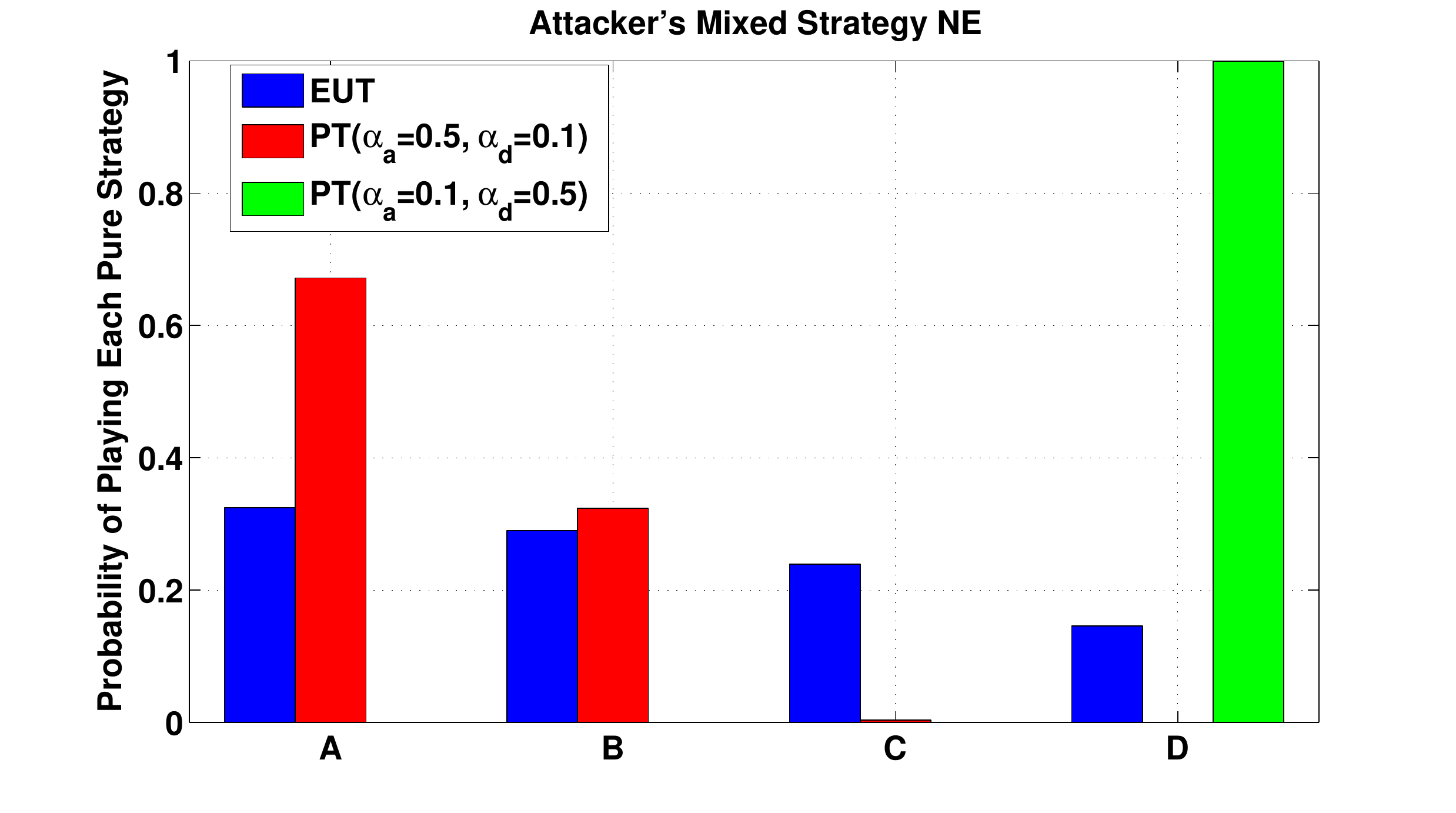}
 \vspace{-0.3cm}
   \caption{\label{fig:Fig6and9} Attacker mixed strategies at the equilibrium for both EUT and PT with i) Scenario 1: $\alpha_a=0.5$ and $\alpha_d=0.1$, and ii) Scenario 2: $\alpha_a=0.1$ and $\alpha_d=0.5$.}
\end{center}\vspace{-0.3cm}
\end{figure}
\begin{figure}[!t]
 \begin{center}
 \vspace{-0.3cm}
  \includegraphics[width=8cm]{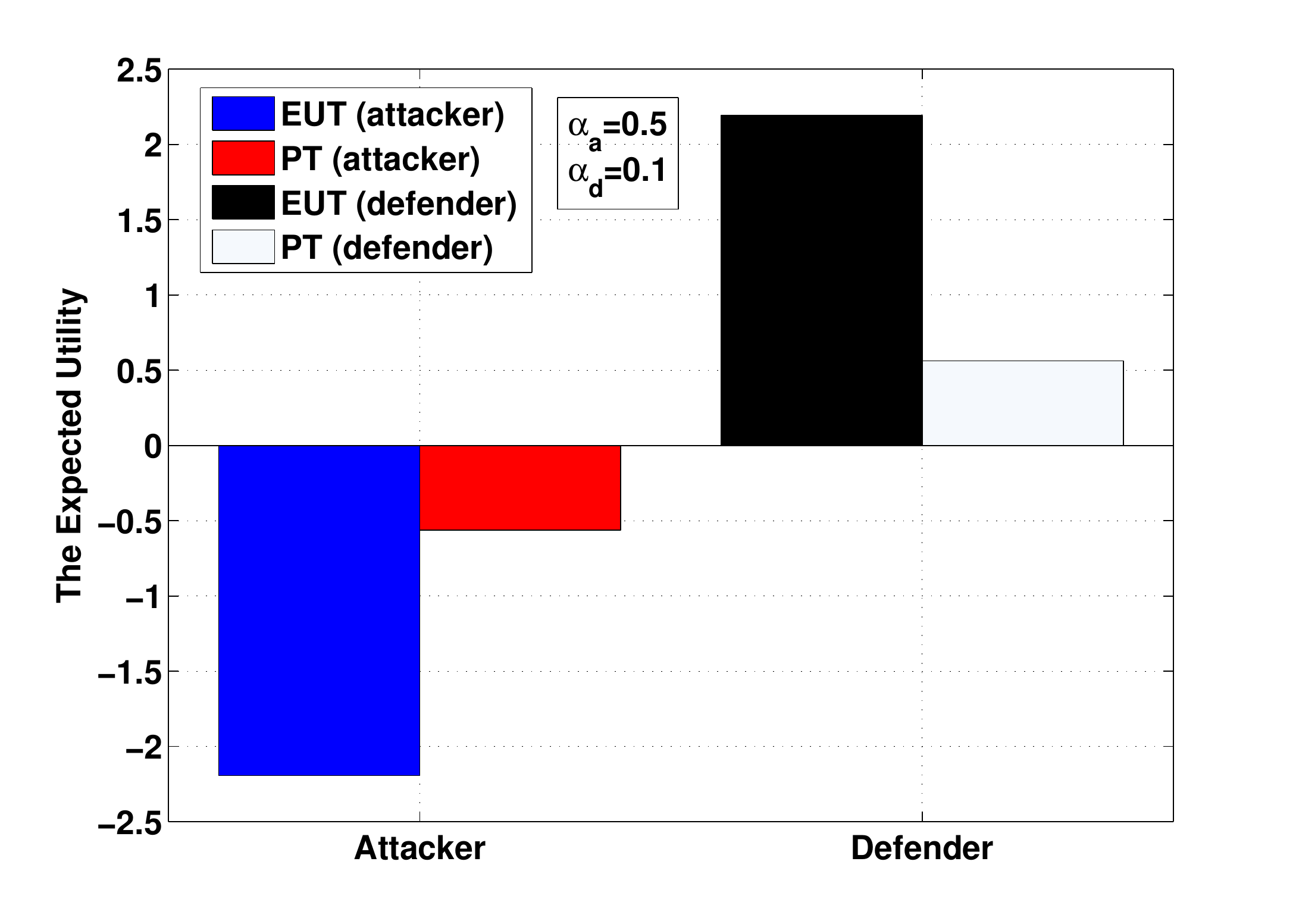}
 \vspace{-0.3cm}
   \caption{\label{fig:ut51} Expected utility at the equilibrium for the attacker and the defender under both EUT and PT with $\alpha_a=0.5$ and $\alpha_d=0.1$ (Scenario 1).}
\end{center}\vspace{-0.3cm}
\end{figure}
\begin{figure}[!t]
 \begin{center}
 \vspace{-0.3cm}
  \includegraphics[width=8cm]{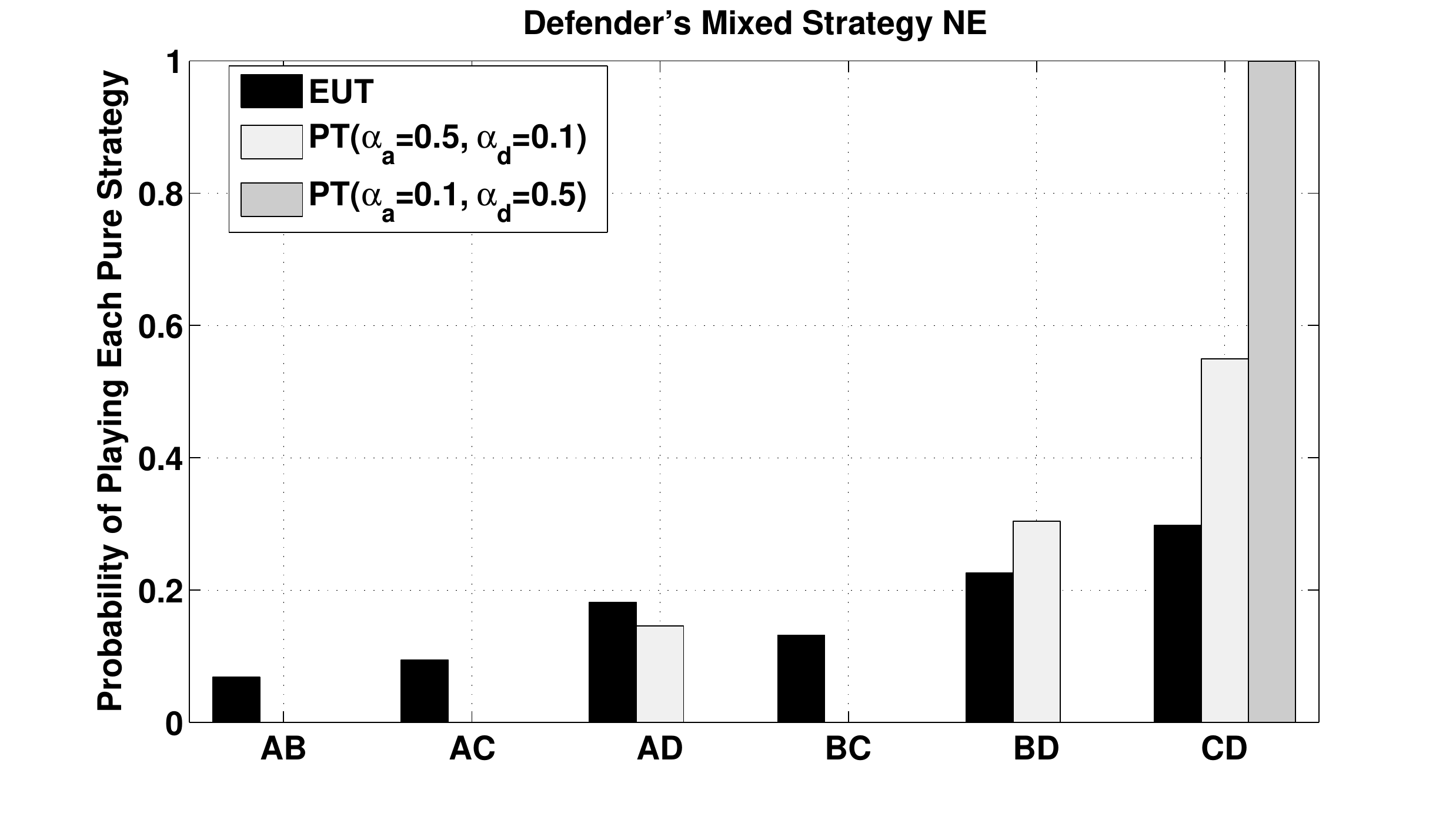}
 \vspace{-0.3cm}
   \caption{\label{fig:Fig7and10} Defender mixed strategies at the equilibrium for both EUT and PT with i) Scenario 1: $\alpha_a=0.5$ and $\alpha_d=0.1$, and ii) Scenario 2: $\alpha_a=0.1$ and $\alpha_d=0.5$.}
\end{center}\vspace{-0.6cm}
\end{figure}
\begin{figure}[!t]
 \begin{center}
 \vspace{-0.3cm}
  \includegraphics[width=8cm]{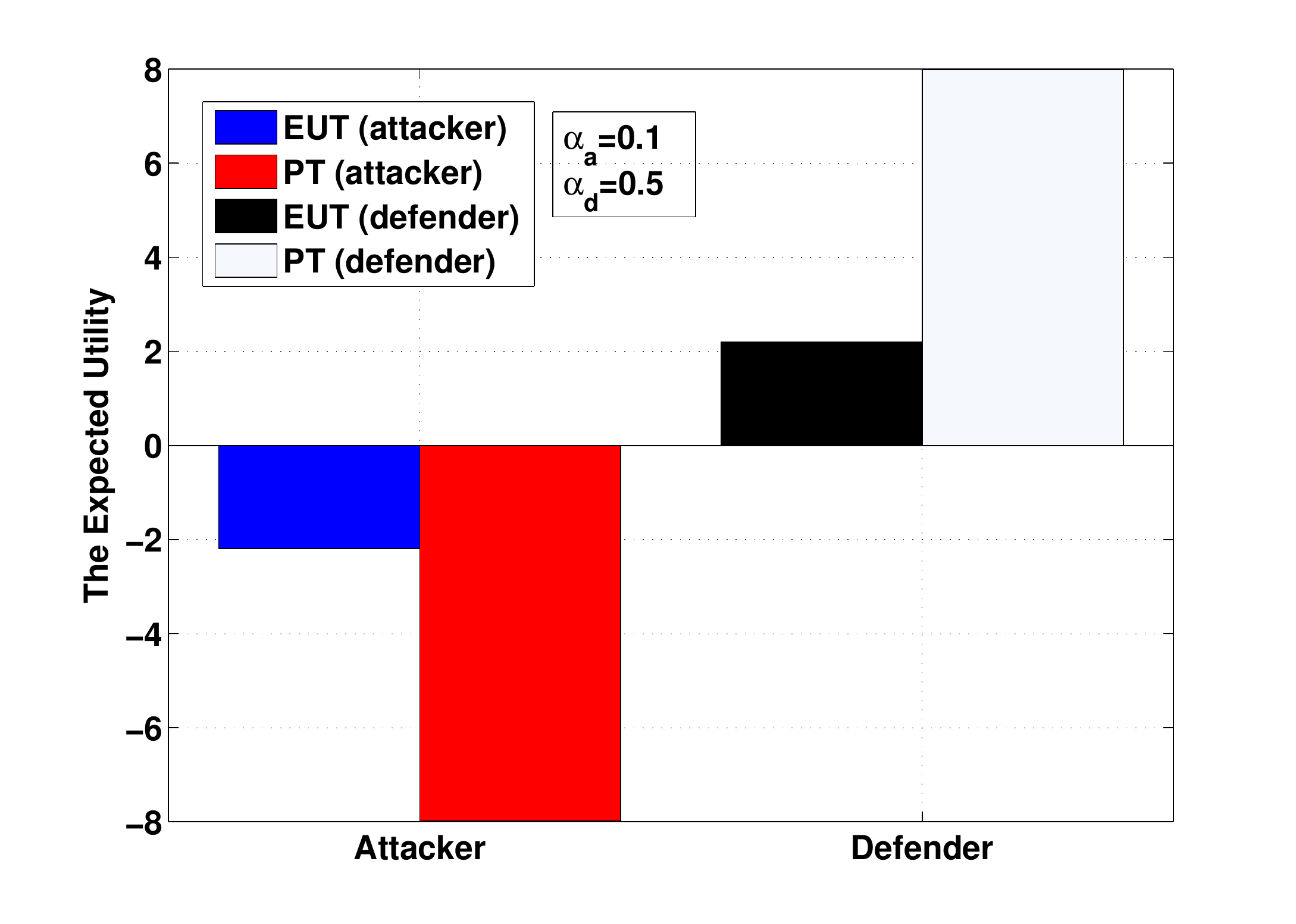}
 \vspace{-0.3cm}
   \caption{\label{fig:ut15} Expected utility at the equilibrium for the attacker and the defender under both EUT and PT with $\alpha_a=0.1$ and $\alpha_d=0.5$ (Scenario 2).}
\end{center}\vspace{-0.3cm}
\end{figure}
In the first scenario, the defender becomes extremely risk-averse and spends all of its resources for testing the combinations pertaining to the most lethal trojan as shown in Fig.~\ref{fig:Fig7and10}. This, in turn, will leave the defender at a disadvantage. Indeed, under such a significantly conservative defense strategy, the attacker finds it less risky to simply mix its attacks between the two trojans with lowest damage $A$ and $B$ as shown in Fig.~\ref{fig:Fig6and9}. By doing so, inadvertently, the attacker will benefit and will become more likely to emerge as a winner in the game. This is demonstrated by the results in Fig.~\ref{fig:ut51} where we can see that the value of the game decreases by about $76\%$.

In the second scenario, by being completely irrational about the perceived defense strategies, under PT, the attacker keeps attempting to insert the most damaging trojan $D$ as shown in Fig.~\ref{fig:Fig6and9}. In contrast, as the defender remains relatively risk averse for $\alpha_d=0.5$, it spends most of its effort to detect the trojans with most damage $CD$ as shown in Fig.~\ref{fig:Fig7and10}. By leveraging its ``rationality'' advantage, the defender can continuously detect the attacker's trojan and, thus, as seen in Fig.~\ref{fig:ut15}, the average value of the game is equal to $8$, which is the value of the fine.

In Fig.~\ref{fig:proValp}, we study the case in which both attacker and defender have an equal rationality parameter, i.e., $\alpha_a=\alpha_d=\alpha$. In this figure, we show the equilibrium mixed-strategy probability for the most damaging strategy $D$ for the attacker and the most defensive strategy $CD$, for the defender. Fig.~\ref{fig:proValp} shows very interesting insights on the trojan hardware detection game. First, for games in which both the defender and the attacker significantly deviate from the rational path ($\alpha < 0.3$), the outcome of the game leads to both players using their most conservative strategies, with probability $1$. This directly implies that, for this highly irrational case, the defender will always emerge as a winner. In contrast, for the regime $0.3 \le \alpha \le 0.7$, under which both the attacker and the defender are not completely rational (but have equal rationality level), the attacker becomes less likely to use its most damaging strategy $D$, as compared to the fully rational case, while the defender becomes more likely to use its most protective strategy $CD$, in comparison to the fully rational case. This naturally translates in an advantage for the attacker (compared to the EUT case) as seen previously in Figs.~\ref{fig:ut5} and \ref{fig:ut51} as well as in Fig.~\ref{fig:utValp}.

\begin{figure}[!t]
 \begin{center}
 \vspace{-0.3cm}
  \includegraphics[width=8cm]{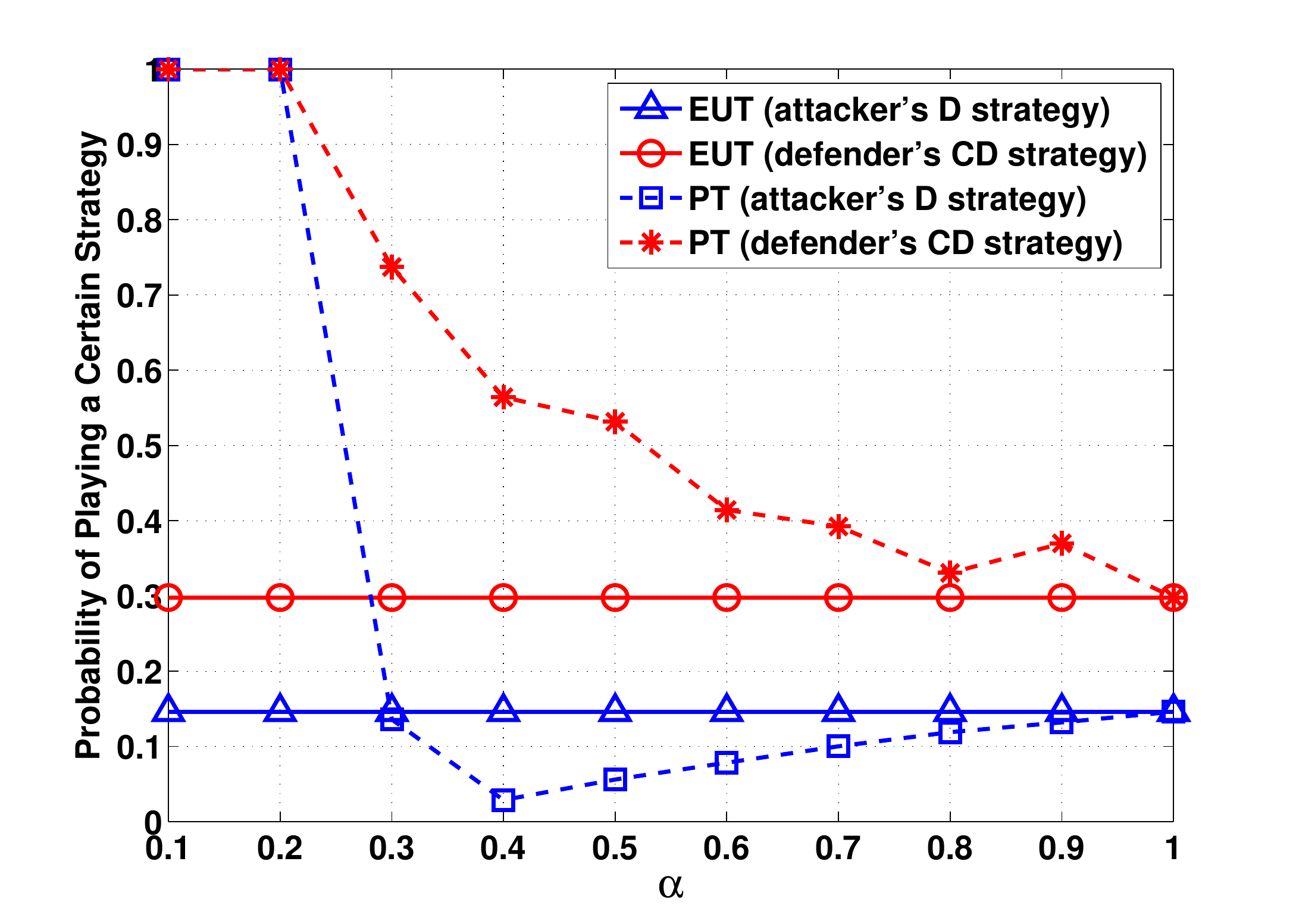}
 \vspace{-0.3cm}
   \caption{\label{fig:proValp} Equilibrium mixed strategies under PT and EUT for the most conservative defender and attacker options as the rationality of both players $\alpha_a=\alpha_d=\alpha$ varies.}
\end{center}\vspace{-0.6cm}
\end{figure}
\begin{figure}[!t]
 \begin{center}
 \vspace{-0.3cm}
  \includegraphics[width=8cm]{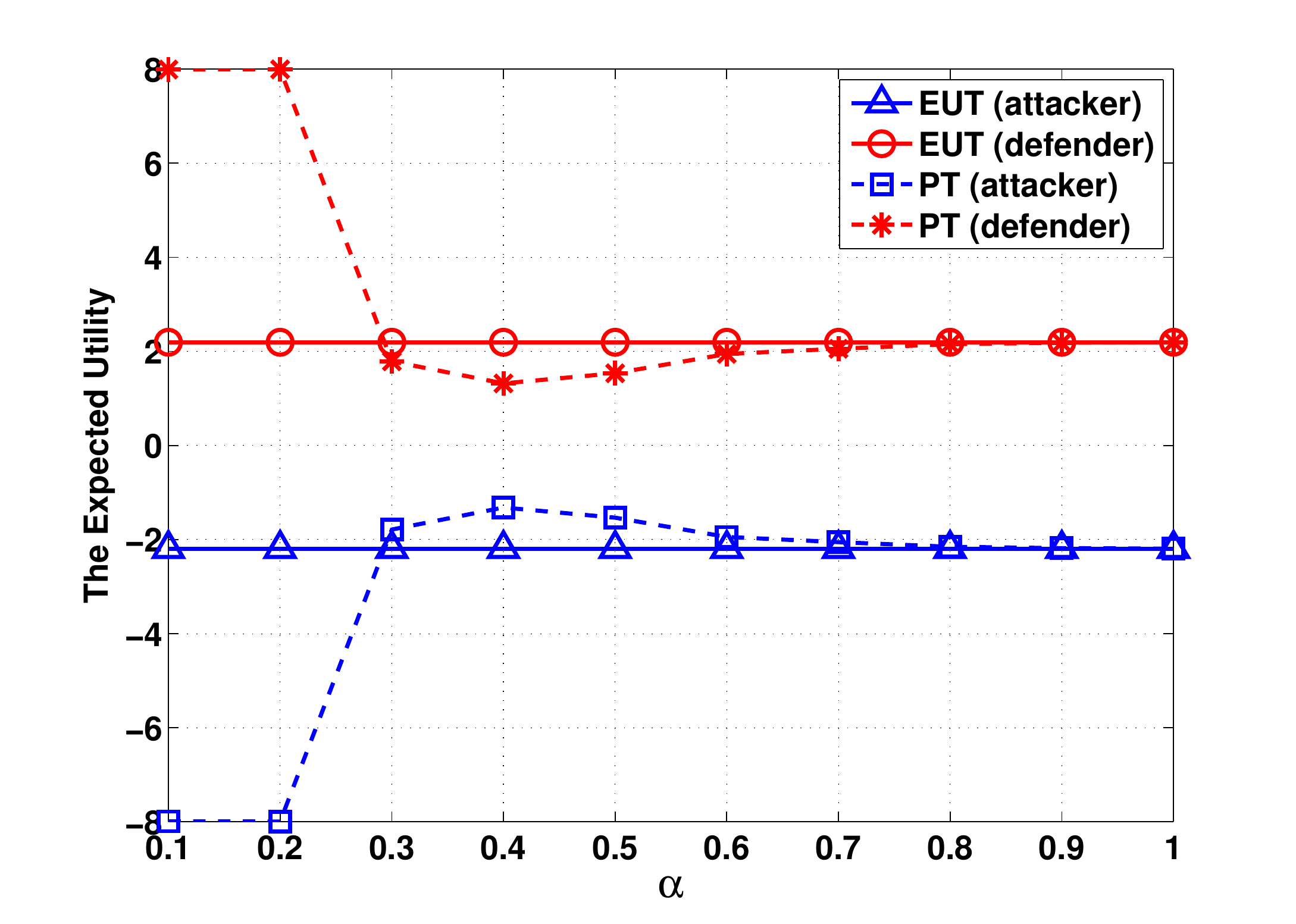}
 \vspace{-0.3cm}
   \caption{\label{fig:utValp} Expected utility at the equilibrium under PT and EUT for the most conservative defender and attacker options as the rationality of both players $\alpha_a=\alpha_d=\alpha$ varies.}
\end{center}\vspace{-0.3cm}
\end{figure}

Fig.~\ref{fig:utValp-atk} shows the expected utility at both the PT and EUT equilibria for a scenario in which the defender is completely rational ($\alpha_d=1$) while the attacker has a varying rationality parameter. Fig.~\ref{fig:utValp-atk} shows that, under a completely rational defense strategy, the EUT performance will upper bound the attacker's performance. In other words, the attacker cannot do better than by behaving somewhat in line with the rational path, as the two utilities coincide for $\alpha_a > 0.3$. Moreover, under a perfectly rational defense strategy, the attacker will immediately be detected if it deviates significantly from the EUT behavior, as evidenced in Fig.~\ref{fig:utValp-atk}  by the expected utility achieved for $\alpha < 0.3$.

In Fig.~\ref{fig:utValp-dfd}, we consider the case in which the attacker is completely rational $\alpha_a=1$ while the defender has a varying rationality level. Fig.~\ref{fig:utValp-dfd} shows that, as the rationality of the defender increases, its defense mechanism performs better. Indeed, by avoiding extremely conservative and irrational perceptions of the attack strategy, i.e., for $\alpha_d \ge 0.4$, the defender can maintain the performance of the system within the bounds of the fully rational EUT behavior even if its own rationality is below that of the attacker. In contrast, for $\alpha_d < 0.4$, the fully rational attacker will be able to exploit its rationality advantage and will thus have better chances of damaging the system. This damage increases with decreasing $\alpha_d$. The worst-case system operation occurs when the defender has a rationality parameter of $\alpha_d \le 0.2$.\vspace{-0.1cm}
\begin{figure}[!t]
 \begin{center}
 \vspace{-0.3cm}
  \includegraphics[width=8cm]{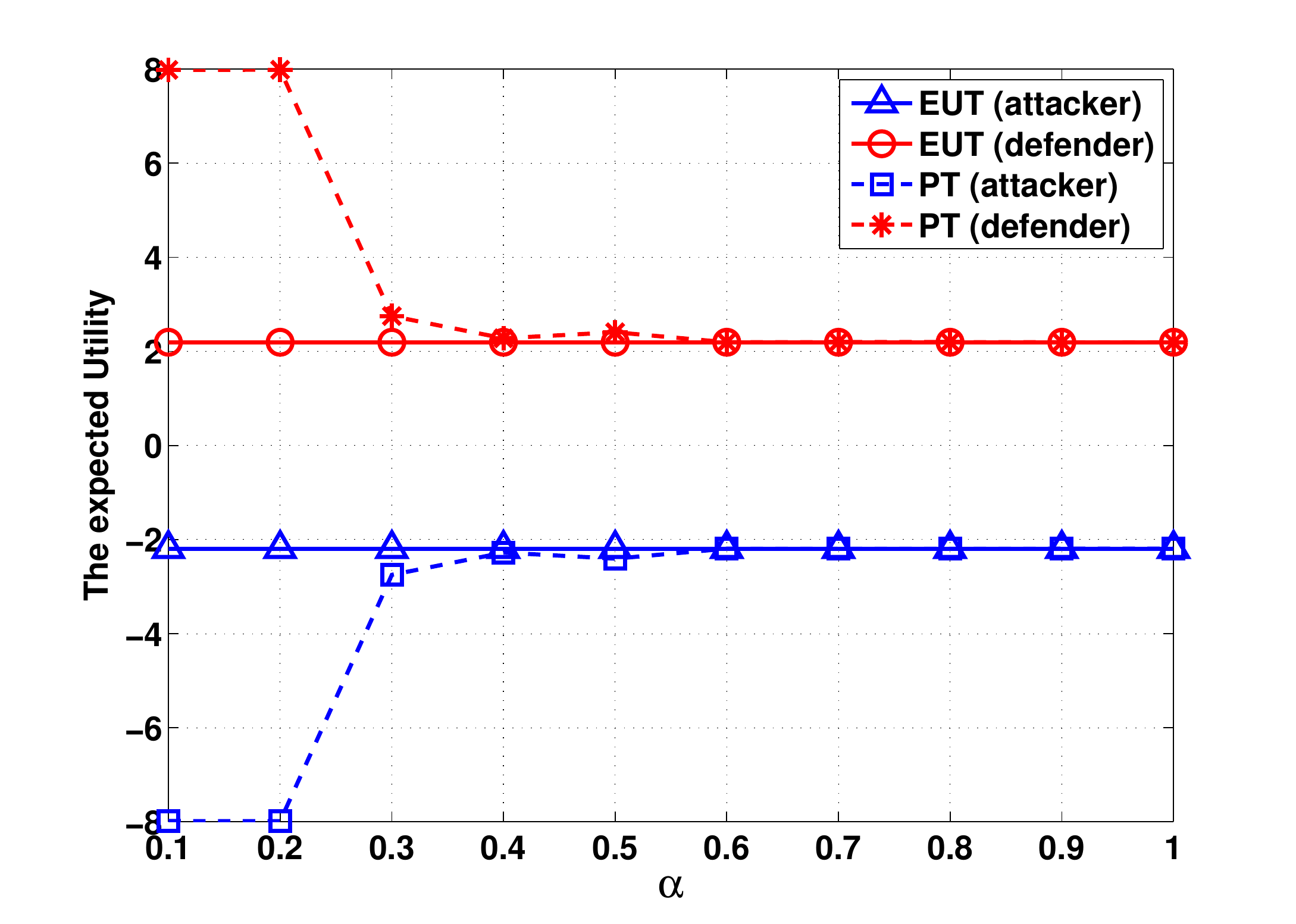}
 \vspace{-0.3cm}
   \caption{\label{fig:utValp-atk} The expected utility at the equilibrium as the rationality of the attacker varies, under a completely rational defender with $\alpha_d=1$.}
\end{center}\vspace{-0.6cm}
\end{figure}
\begin{figure}[!t]
 \begin{center}
 \vspace{-0.3cm}
  \includegraphics[width=8cm]{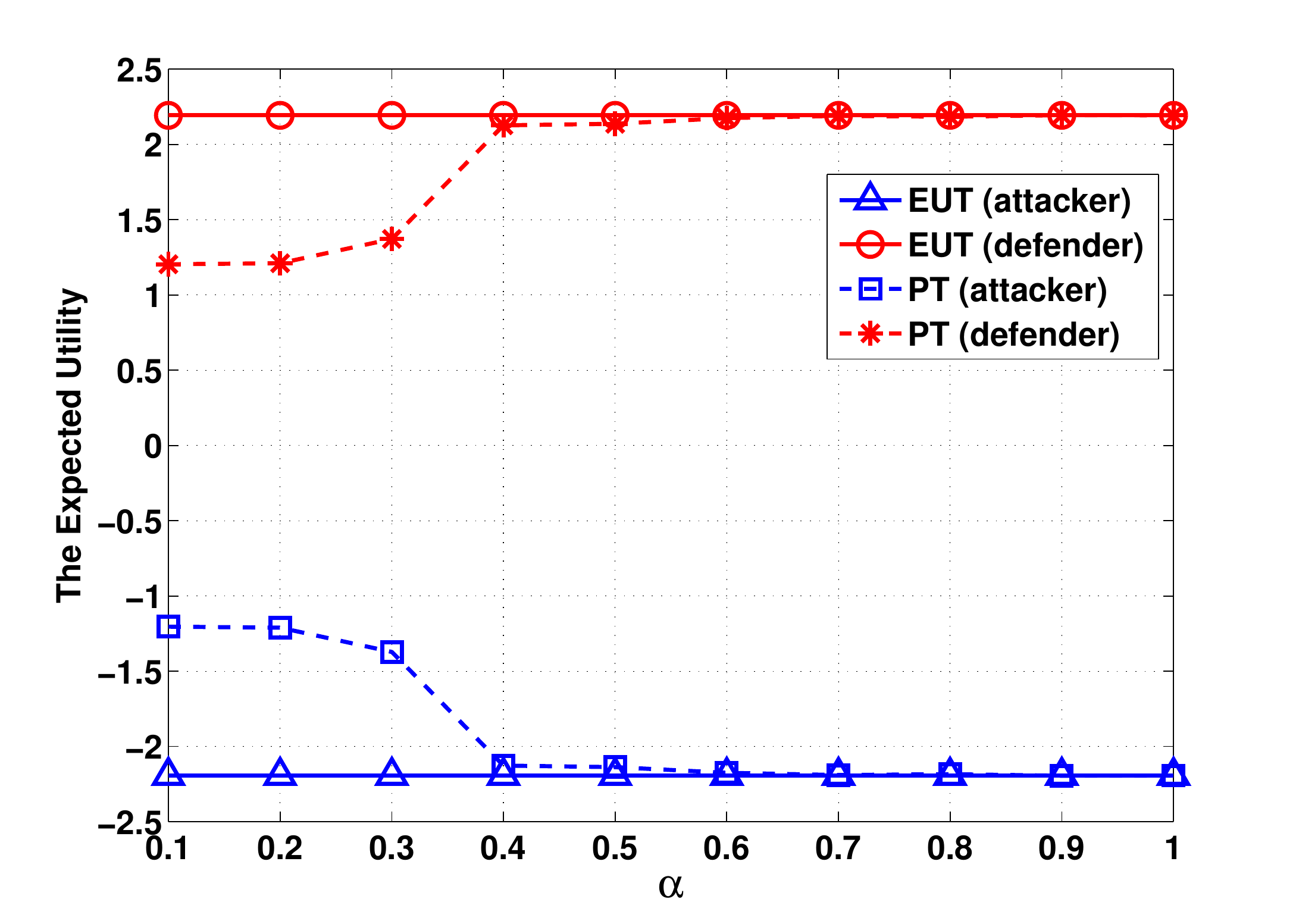}
 \vspace{-0.3cm}
   \caption{\label{fig:utValp-dfd} The expected utility at the equilibrium as the rationality of the defender varies, under a completely rational attacker with $\alpha_a=1$.}
\end{center}\vspace{-0.3cm}
\end{figure}

Here, we note that the Nash equilibrium strategies for the attacker and the defender under both EUT and PT which have been introduced and analyzed in this section were obtained using the proposed solution algorithm in Table~\ref{alg:alg1}. In all of the studied cases, the algorithm has successfully converged in a relatively short period of time.
In fact, by inspecting the algorithm in Table 1, one can see that the most computationally demanding operation of the algorithm is Step 5 in which each player chooses the action that maximizes its payoff given the perceived empirical frequencies of the actions of the opponent. Given that the action space of each of the players is discrete, this consists of searching over all the elements of each player's action space. This search, however, requires very low computational complexity which grows linearly with the size of the action spaces of each player. As a result, computing Step 5 at each iteration requires a very short amount of time. All other needed computations in Table~\ref{alg:alg1}, steps 6-10, are simple algebraic computations requiring a very short execution time. Hence, the execution time of each iteration of the algorithm, and as a result its total convergence time, is practically very short. 

For instance, Figs.~\ref{fig:AttEUT}--\ref{fig:DefPT} show the convergence of the strategies of the attacker and defender to the NE under EUT and PT, respectively, for the case treated in Figs.~\ref{fig:atk5}-\ref{fig:ut5} with $\alpha_a=\alpha_d=0.5$. In these set of simulations, the stoppage criterion as defined in~(\ref{eq:ConvergenceIt}) and in Step 8 of Table~\ref{alg:alg1}, is chosen such that $\frac{1}{M}=0.001$. In other words, the algorithm is considered to have converged when the change in the updated empirical frequencies of all the actions of both players is less than 0.001.  
\begin{figure}[!t]
 \begin{center}
 \vspace{-0.3cm}
  \includegraphics[width=8cm]{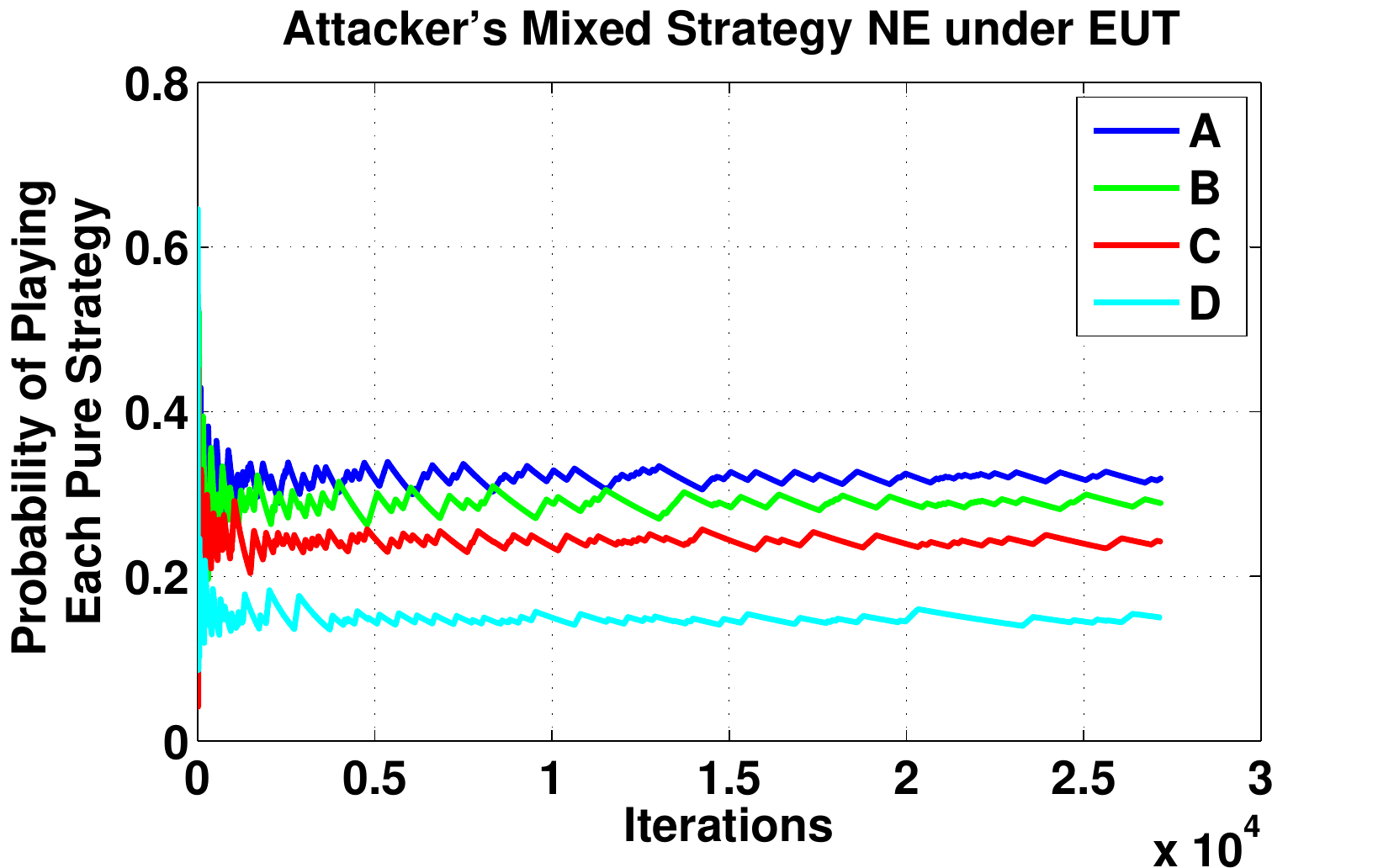}
 %\vspace{-0.3cm}
   \caption{\label{fig:AttEUT} Convergence of the attacker's NE strategies under EUT.}
\end{center}\vspace{-0.6cm}
\end{figure}
\begin{figure}[!t]
 \begin{center}
 \vspace{-0.3cm}
  \includegraphics[width=8cm]{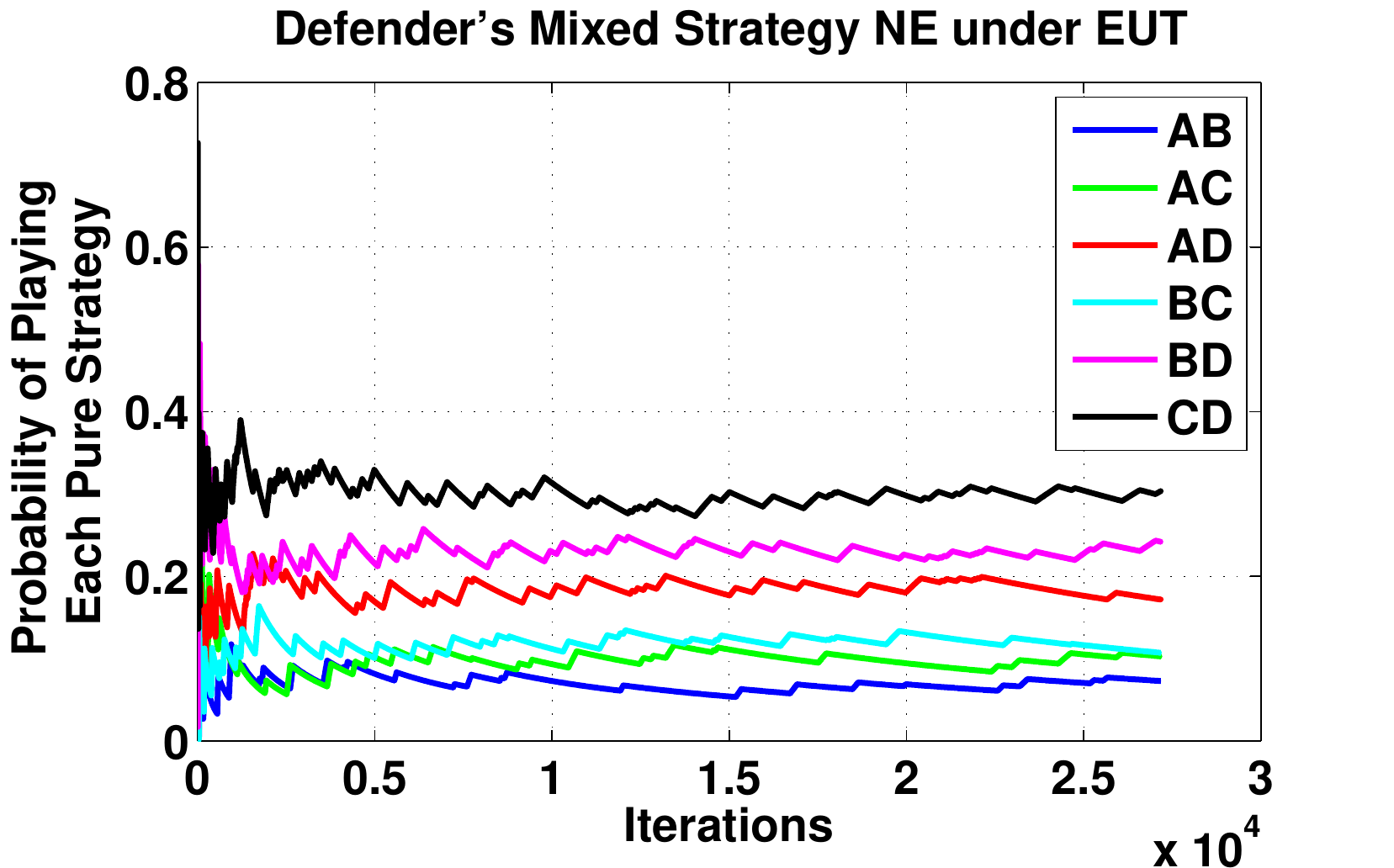}
 %\vspace{-0.3cm}
   \caption{\label{fig:DefEUT} Convergence of the defender's NE strategies under EUT.}
\end{center}\vspace{-0.3cm}
\end{figure} 
Fig.~\ref{fig:AttEUT} and Fig.~\ref{fig:DefEUT} show, respectively, the convergence of the attacker's and defender's mixed strategies under EUT while Fig.~\ref{fig:AttPT} and Fig.~\ref{fig:DefPT} show, respectively, the convergence of the attacker's and defender's mixed strategies under PT. Here, we note that even though the number of iterations needed for the attacker's and defender's strategies to converge is relatively high, it only took the algorithm 26.1 seconds to converge in the case of EUT (i.e. Fig.~\ref{fig:AttEUT} and~\ref{fig:DefEUT}) and 26.4 seconds to converge in case of PT (i.e. Fig.~\ref{fig:AttPT} and Fig~\ref{fig:DefPT}) using a 2.53 GHz processor and 3GB RAM computer. Here, we note that the convergence required a large number of iterations due to two main reasons: i) the very small convergence criterion $\frac{1}{M}$ that we have chosen, and ii) the decreasing influence of each iteration when the number of iterations grows large. In fact, as can be seen from Step 6 in Table~\ref{alg:alg1}, as the number of iterations $k$ increases, the effect of each iteration on updating the empirical frequency decreases. As such, for a very small convergence criterion, it would require the algorithm a large number of iterations to converge. However, since the computational requirement of the algorithm is very low, the execution of each iteration takes a very short time. Hence, as can be seen from our generated results, even though the algorithm required a large number of iterations, the total convergence time is kept practically small. Moreover, an operator can increase $\frac{1}{M}$, if needed, in order to have a faster convergence, at the expense of reaching an approximate rather than exact equilibrium point.

\begin{figure}[!t]
 \begin{center}
 \vspace{-0.3cm}
  \includegraphics[width=8cm]{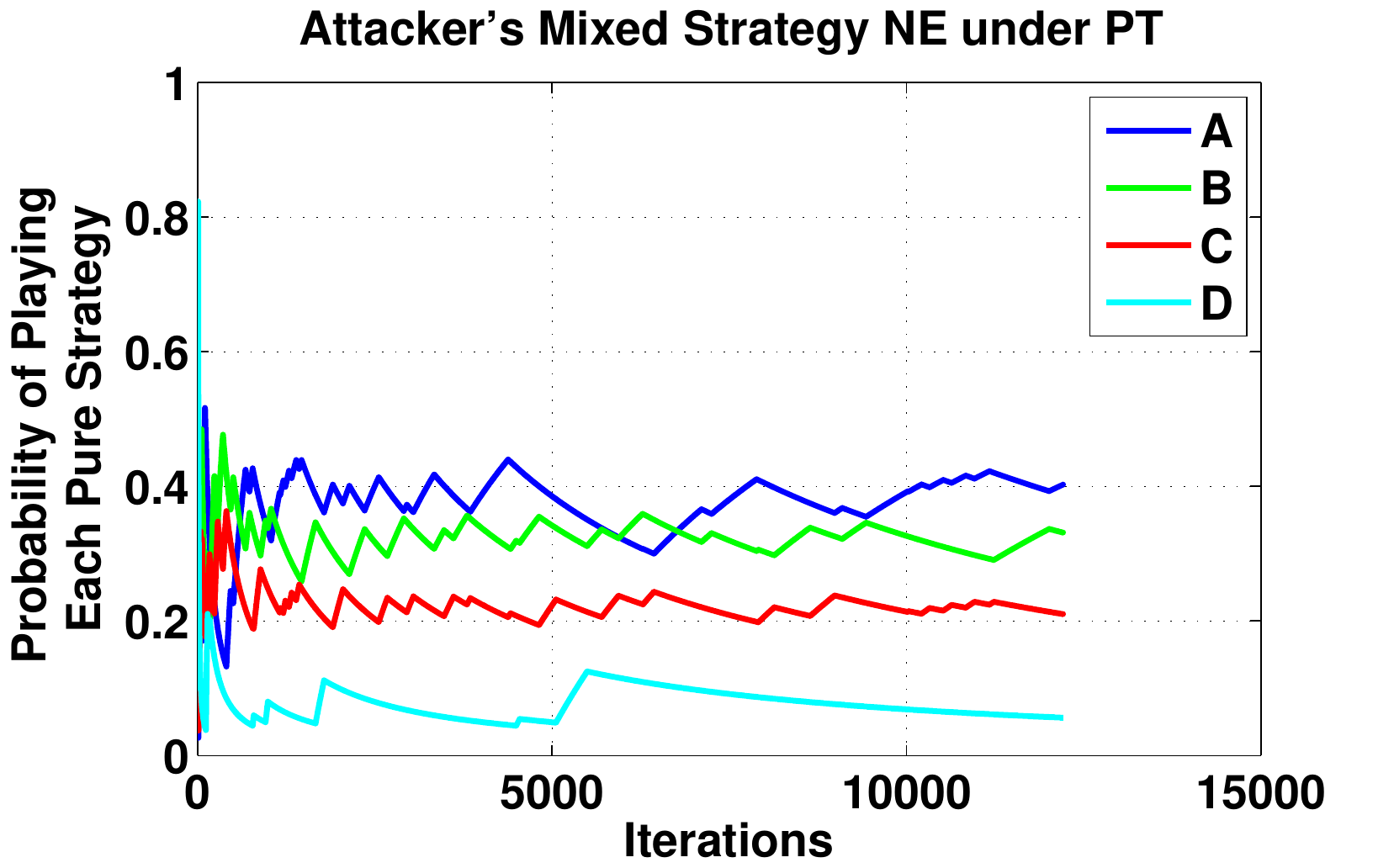}
 %\vspace{-0.3cm}
   \caption{\label{fig:AttPT} Convergence of the attacker's NE strategies under PT.}
\end{center}\vspace{-0.3cm}
\end{figure}
\begin{figure}[!t]
 \begin{center}
 \vspace{-0.3cm}
  \includegraphics[width=8cm]{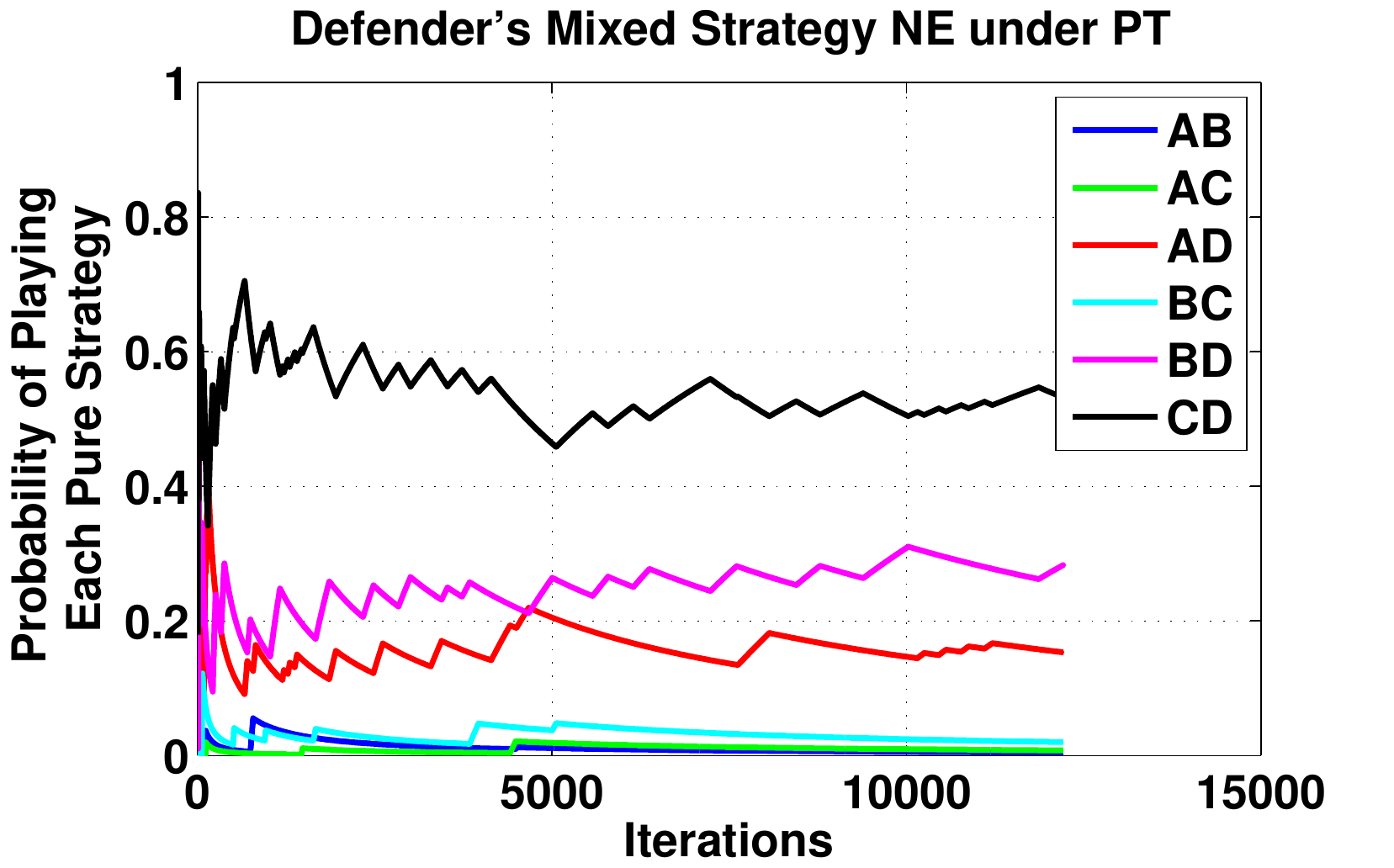}
 %\vspace{-0.3cm}
   \caption{\label{fig:DefPT} Convergence of the defender's NE strategies under PT.}
\end{center}\vspace{-0.6cm}
\end{figure}

\section{Conclusions}\label{sec:conc}%\vspace{-0.2cm}
In this paper, we have proposed a novel game-theoretic approach for modeling the interactions between hardware manufacturers, who can act as attackers by inserting hardware trojans, and companies or agencies, that act as defenders that test the circuits for hardware torjans. We have formulated the problem as a noncooperative game between the attacker and the defender, in which the attacker chooses the optimal trojan type to insert while the defender chooses the best testing strategy, from a set of trojan types. To account for the uncertainty and risk in the decision making processes, we have proposed a novel framework, based on the emerging tools of prospect theory, for analyzing the proposed game. To solve the game for both conventional game theory and for prospect theory, we have proposed a fictitious play-based algorithm and shown its guaranteed convergence to an equilibrium point. Thorough analytical and simulation results have been derived to assess the outcomes of the proposed games. Our results have shown that the use of prospect-theoretic considerations can provide insightful information on how irrational behavior, uncertainty, and risk can impact the interactions between an attacker and defender in a hardware trojan detection game. 
\bibliographystyle{IEEEtran}
\bibliography{references}
\appendix[Proof of Theorem~\ref{th:FEUT}]
Starting first with the EUT case, since for the case in which no player is a winner the expected utility of the attacker is equal to that of the defender and the MSNE strategies of the attacker are unique, we solve for $F$ from the perspective of the attacker's MSNE. In particular,\\
\begin{align}
\begin{cases}
&U_a^{\text{EUT}}(\boldsymbol{p}_d^{*\text{EUT}},\boldsymbol{p}_a^{*\text{EUT}})={\boldsymbol{p}_a^{*\text{EUT}}}'\cdot  \boldsymbol{M}_d \cdot \boldsymbol{p}_d^{*\text{EUT}}=0,\\
&U_d^{\text{EUT}}(\boldsymbol{p}_d^{*\text{EUT}},\boldsymbol{p}_a^{*\text{EUT}})={\boldsymbol{p}_d^{*\text{EUT}}}'\cdot  \boldsymbol{M}_a \cdot \boldsymbol{p}_a^{*\text{EUT}}=0,\\
\end{cases}
\end{align}%\vspace{0.2cm}
\\
where ${\boldsymbol{p}_a^{*\text{EUT}}}'$ is the transpose of $\boldsymbol{p}_a^{*\text{EUT}}$. Here, the expected utility of the defender requires one to first compute the MSNE of the attacker using $\boldsymbol{M}_a$. Based on the indifference principle, at the defender's MSNE, we have $U_d(AB, \boldsymbol{p}_a^*)=U_d(AC, \boldsymbol{p}_a^*)=\cdots=U_d(CD, \boldsymbol{p}_a^*)$. Moreover, we have:\\
\begin{align}
[U_d(AB, \boldsymbol{p}_a^*)\  U_d(AC, \boldsymbol{p}_a^*)\   \cdots  U_d(CD, \boldsymbol{p}_a^*)]^T=\boldsymbol{M}_a\cdot \boldsymbol{p}_a^{*\text{EUT}}
\end{align}\\\indent
Because the mixed strategy of the defender is nonnegative, i.e. $\boldsymbol{p}_d^{*\text{EUT}} \ge 0$, we have
\begin{align}
\begin{split}
{\boldsymbol{p}_d^{*\text{EUT}}}'\cdot \boldsymbol{U}_d(\boldsymbol{s}_d,\boldsymbol{p}_a^*)=&0\\
\boldsymbol{U}_d(\boldsymbol{s}_d,\boldsymbol{p}_a^*)=&0\\
\therefore \quad \boldsymbol{M}_a\cdot \boldsymbol{p}_a^{*\text{EUT}}=&\boldsymbol{0}.
\end{split}
\end{align}\\\indent
In particular, for $U_d(AB,\boldsymbol{p}_a^*)$,\\
\begin{align}\label{eq:FEUT}
\begin{split}
&Fp_a^*(A)+Fp_a^*(B)-4p_a^*(C)-12p_a^*(D)=0,\\
&F^v_{\textrm{EUT}}=\frac{4p_a^*(C)+12p_a^*(D)}{p_a^*(A)+p_a^*(B)}.
\end{split}
\end{align}\vspace{0.1cm}

For the case of PT, similarly to the case of EUT, we have\\
\begin{align}
\begin{cases}
&U_a^{\text{PT}}(\boldsymbol{p}_d^{*\text{PT}},\boldsymbol{p}_a^{*\text{PT}})={\boldsymbol{p}_a^{*\text{PT}}}'\cdot  \boldsymbol{M}_d\cdot \boldsymbol{p}_d^{*\text{PT}}=0,\\
&U_d^{\text{PT}}(\boldsymbol{p}_d^{*\text{PT}},\boldsymbol{p}_a^{*\text{PT}})={\boldsymbol{p}_d^{*\text{PT}}}'\cdot  \boldsymbol{M}_a\cdot \boldsymbol{p}_a^{*\text{PT}}=0.\\
\end{cases}
\end{align}\\\indent
Although, at the the mixed NE the indifference principle holds, $\boldsymbol{M}_a \cdot \boldsymbol{p}_a^{*\text{PT}} \neq \boldsymbol{0}$ due to the nonlinear weighting effect. Thus,
\begin{align}
F^v_{\textrm{PT}}=\frac{{\boldsymbol{p}_a^{*\text{PT}}}'\cdot
\begin{bmatrix}
0&0&0& 1& 1&1\\
0&2 &2 &0&0&2\\
4 &0&4 &0&4&0\\
12&12&0&12&0&0\\
\end{bmatrix}
\cdot \boldsymbol{p}_d^{*\text{PT}}}{{\boldsymbol{p}_a^{*\text{PT}}}'\cdot
\begin{bmatrix}
1&1&1& 0& 0&0\\
1&0 &0 &1&1&0\\
0 &1&0 &1&0&1\\
0&0&1&0&1&1\\
\end{bmatrix}
\cdot \boldsymbol{p}_d^{*\text{PT}}}
\end{align}\\\indent
Since the denominator is not $0$, then $F^v_{\textrm{PT}}$ can be computed.

\end{document}